\documentclass[journal]{IEEEtran}
\usepackage[]{graphicx}\usepackage[]{color}
\makeatletter
\def\maxwidth{ %
  \ifdim\Gin@nat@width>\linewidth
    \linewidth
  \else
    \Gin@nat@width
  \fi
}
\makeatother

\definecolor{fgcolor}{rgb}{0.345, 0.345, 0.345}

\usepackage{framed}
\makeatletter
 {\par\unskip\endMakeFramed%
 \at@end@of@kframe}
\makeatother

\definecolor{shadecolor}{rgb}{.97, .97, .97}
\definecolor{messagecolor}{rgb}{0, 0, 0}
\definecolor{warningcolor}{rgb}{1, 0, 1}
\definecolor{errorcolor}{rgb}{1, 0, 0}

\usepackage{alltt}

\usepackage[utf8]{inputenc}

\usepackage[normalem]{ ulem }
\usepackage{soul}
\usepackage{xr-hyper}

\makeatletter
\newcommand*{\addFileDependency}[1]{
  \typeout{(#1)}
  \@addtofilelist{#1}
  \IfFileExists{#1}{}{\typeout{No file #1.}}
}
\makeatother

\usepackage{cite}   
\usepackage{bm} 
\usepackage{amsthm}  
\usepackage{color}
\usepackage{amssymb}  
\usepackage{multirow}
\usepackage{amsbsy}  
\usepackage{amsfonts} 
\usepackage{float}
\usepackage{amsmath}
\usepackage{booktabs}
\usepackage[font=small, labelfont=sc, textfont = it]{caption}  
\usepackage{nicefrac} 
\usepackage[shortlabels]{enumitem}
\usepackage{bbm}
\usepackage{subfigure}
\usepackage{comment}
\usepackage{graphicx}

\usepackage[hyperindex=true,pdftitle={},
pdfauthor={Gaetan Bakalli},colorlinks=TRUE,
pagebackref=false,citecolor=blue,plainpages=false,
pdfpagelabels]{hyperref}

\usepackage{mathtools}

\usepackage[boxed]{algorithm2e}

\DeclareMathOperator*{\var}{var}  
\DeclareMathOperator*{\tr}{tr}
\DeclareMathOperator*{\cov}{Cov}

\DeclareMathOperator*{\argmin}{argmin} 
\DeclareMathOperator*{\argzero}{argzero}

\DeclareMathOperator*{\Int}{Int}
\def\real{{\rm I\!R}}

\def\bvartheta{\bm{\vartheta}}
\def\btheta{\bm{\theta}}
\def\hbtheta{\hat{\btheta}}
\def\bnu{\bm{\nu}}
\def\hbnu{\hat{\bnu}}

\usepackage{graphicx}

\def\boxit#1{\vbox{\hrule\hbox{\vrule\kern3pt
          \vbox{\kern3pt#1\kern3pt}\kern3pt\vrule}\hrule}}

\newtheoremstyle{mytheoremstyle} 
    {0.3cm}                      
    {0cm}                        
    {\itshape}                   
    {}                           
    {\scshape}                   
    {: }                          
    {0em}                       
    {}  

\theoremstyle{mytheoremstyle}
\newtheorem{Theorem}{Theorem}
\newtheorem{Lemma}{Lemma}
\newtheorem{Corollary}{Corollary}
\newtheorem{Assumption}{Assumption}
\newtheorem{Proposition}{Proposition}
\newtheorem{Remark}{Remark}

\renewenvironment{proof}{{\noindent \sc Proof:}}{\qed}

\newtheoremstyle{myExampleRemarkstyle} 
    {0.3cm}                    
    {0cm}                           
    {\itshape}                   
    {}                           
    {\scshape}                   
    {: }                          
    {0em}                       
    {}  

\theoremstyle{myExampleRemarkstyle}
 
\newtheoremstyle{simuStyle}
{0.3cm} 
{0cm} 
{} 
{} 
{\bfseries} 
{.} 
{0em} 
{} 

\theoremstyle{simuStyle}

\newtheoremstyle{stratStyle}
{0.3cm} 
{0cm} 
{} 
{} 
{\scshape} 
{: } 
{0em} 
{} 

\theoremstyle{stratStyle}

\DeclareSymbolFont{lettersA}{U}{txmia}{m}{it}
\DeclareMathSymbol{\field}{\mathord}{lettersA}{"83}
\IfFileExists{upquote.sty}{\usepackage{upquote}}{}

\begin{document}

\title{Multi-Signal Approaches for Repeated Sampling Schemes in Inertial Sensor Calibration}

\author{Gaetan~Bakalli,
		Davide~A.~Cucci,
        Ahmed~Radi,
        Naser~El-Sheimy,
		Roberto~Molinari,
		Olivier Scaillet and~St\'ephane~Guerrier

\thanks{G. Bakalli is with the Department of Mathematics \& Statistics, Auburn University, Auburn, AL 36849, USA (e-mail: gaetan.bakalli@auburn.edu).}
\thanks{D. Cucci is with the Geneva School of Economics and Management, University of Geneva, 1205, Switzerland (e-mail: davide.cucci@unige.ch).}
\thanks{A. Radi is with the Technical Researches Center, Cairo, Egypt (e-mail: ahmed.elboraee@ucalgary.ca).}
\thanks{N. El-Sheimy is with the Department of Geomatics Engineering, University of Calgary, Calgary, Alberta T2N 1N4, Canada (e-mail: elsheimy@ucalgary.ca).}
\thanks{R. Molinari is with the Department of Mathematics \& Statistics, Auburn University, Auburn, AL 36849, USA (e-mail: robmolinari@auburn.edu).}
\thanks{O. Scaillet is with the Geneva Finance Research Institute, University of Geneva and Swiss Finance Institute, Geneva 1211, Switzerland (e-mail: olivier.scaillet@unige.ch).}
\thanks{S. Guerrier is with the Faculty of Science \& Geneva School of Economics and Management, University of Geneva, 1205, Switzerland.  (e-mail: stephane.guerrier@unige.ch).}}

\maketitle
\begin{abstract}
Inertial sensor calibration plays a progressively important role in many areas of research among which navigation engineering. By performing this task accurately, it is possible to significantly increase general navigation performance by correctly filtering out the deterministic and stochastic measurement errors that characterize such devices. While different techniques are available to model and remove the deterministic errors, there has been considerable research over the past years with respect to modelling the stochastic errors which have complex structures. In order to do the latter, different replicates of these error signals are collected and a model is identified and estimated based on one of these replicates. While this procedure has allowed to improve navigation performance, it has not yet taken advantage of the information coming from all the other replicates collected on the same sensor. However, it has been observed that there is often a change of error behaviour between replicates which can also be explained by different (constant) external conditions under which each replicate was taken. Whatever the reason for the difference between replicates, it appears that the model structure remains the same between replicates but the parameter values vary. In this work we therefore consider and study the properties of different approaches that allow to combine the information from all replicates considering this phenomenon, confirming their validity both in simulation settings and also when applied to real inertial sensor error signals. By taking into account parameter variation between replicates, this work highlights how these approaches can improve the average navigation precision as well as obtain reliable estimates of the uncertainty of the navigation solution.
\end{abstract}

\begin{IEEEkeywords}
Generalized Method of Wavelet Moments, Inertial Sensor Calibration, Stochastic Error, Extended Kalman Filter, Navigation
\end{IEEEkeywords}

%
\IEEEpeerreviewmaketitle

\section{Introduction}
\label{sec:intro}
%
%
%
%

\IEEEPARstart{I}{nertial} sensors are ubiquitous in modern navigation systems, with applications ranging from space missions, aviation and drones, to personal navigation in smartphones. They provide high-frequency and short-term precise information on the orientation and velocity change of the platform they are placed on. Inertial measurements are typically integrated with other sources to obtain estimates of the platform position and orientation in space. Examples are Global Navigation Satellite Systems within strap-down inertial navigation~\cite{titterton2004strapdown} and cameras for visual-inertial systems~\cite{huang2019visual}.

Inertial sensors, like any other sensor, have errors that are both deterministic and stochastic. Deterministic errors such as the stable parts of biases, scale factors and non-orthogonality of the axes can be pre-calibrated and removed from the measurements directly. The additive stochastic part of error can only be taken into account ``on-flight'' within the estimation process to serve two main purposes: i) estimation of the time-correlated part of those stochastic errors (to remove them from the measurements and improve navigation accuracy~\cite{wall2005characterization}) and, ii) estimation of uncertainty associated with the navigation states, such as position and orientation.  This requires proper modeling of the stochastic errors of the sensors, often referred to as ``stochastic calibration''.
This task is generally performed in a black-box fashion on a device-per-device basis, acquiring long series of static measurements which are composed by the stochastic error itself, plus constant terms such as gravity and the Earth rotation rate which can be easily removed. Stochastic calibration of inertial sensors has been widely studied in the last decades and various methods are available for this, going from power spectral density analysis \cite{allen1993performance,board1998ieee} to the correlation of filtered sensor outputs \cite{yuksel2010error}. The majority of these methods aim at decomposing these stochastic signals and/or performing system identification procedures to model them \cite{claus1993multiscale, johansson1999stochastic}. The most commonly employed techniques are, for example, those based on Maximum-Likelihood Estimation \cite{nikolic2015maximum,yuksel2011notes} (hereinafter MLE) or the Allan variance~\cite{allan1966statistics} (hereinafter AV), the latter having been initially conceived for the characterization of phase and frequency instability of precision oscillators. The AV approach consists in multiple separate regressions on the linear segments of the AV plots in order to recover the underlying parameters of interest for the stochastic error signal and represents the de-facto standard for inertial sensor stochastic modeling~\cite{ieee1998ieee}. For a detailed discussion, see~\cite{el2007analysis}.
However, the AV plot is a graphical device that requires a manual inspection and is consequently sensitive to the user's proficiency as well as being burdened with many theoretical limitations including significant (asymptotic) bias in the estimated parameters of the postulated stochastic model~\cite{guerrier2020wavelet, guerrier2016theoretical}. 

To overcome the limitations of the AV approach as well as the important computational limitations of MLE techniques, the Generalized Method of Wavelet Moments~(GMWM) was proposed in~\cite{guerrier2013wavelet} and makes use of the quantity called Wavelet Variance (WV) that, in specific settings, is equivalent to the AV up to a constant. Using a matching technique, the WV allows to easily recover the parameters of the postulated stochastic model providing a statistically appropriate and computationally feasible technique for stochastic calibration of inertial sensors.
More in detail, the intuition behind the GMWM can be described via Fig.~\ref{fig:intro_plot} which represents the log-log plot of the estimated Haar WV (equal to the AV up to a constant) for a simulated stochastic error signal (blue line) with its $95\%$ confidence intervals (shaded light-blue area). The first scales of the WV are driven by a White Noise (WN) process (or Angular Random Walk, in gyroscopes), while the elbow between scales $2^{10}$ to $2^{13}$ is mainly defined by an Auto-Regressive process of order 1 (AR1) which consists in a reparametrization of a Gauss-Markov process. Finally, the larger scales highlight the non-stationary processes, in this case given by a Random Walk (RW) or Rate Random Walk. It can be seen how the individual processes contribute to shaping the WV and the idea of the GMWM (defined more formally further on) is to use this shape and the estimated WV to find the underlying processes and relative parameters by minimizing the distance between the estimated WV and the theoretical one implied by the model.
\begin{figure}
    \centering
    \includegraphics[width=\linewidth]{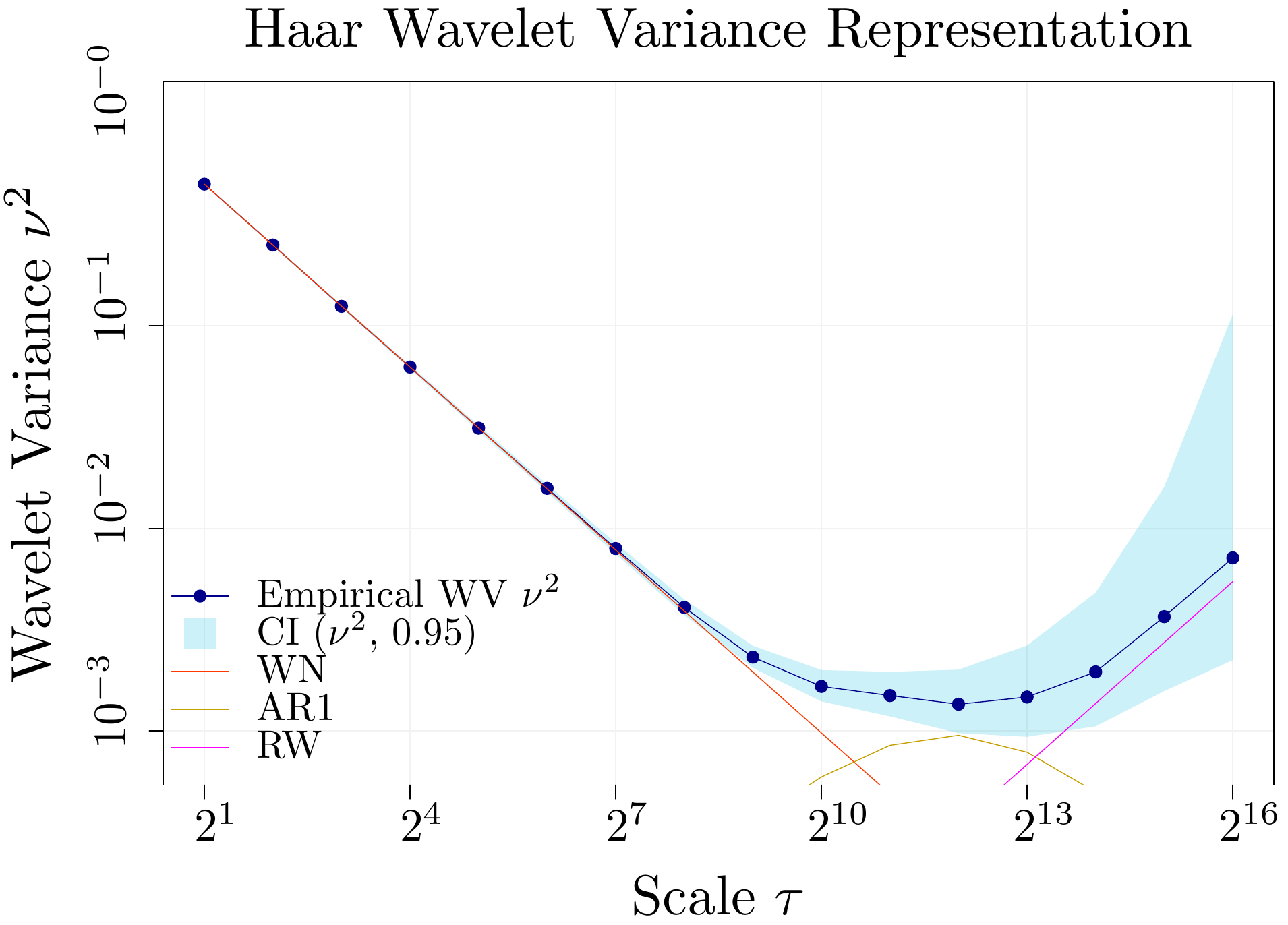}
    \caption{Empirical WV representation (plain blue dotted lines) coming from a synthetic signal simulated from the sum of an WN, AR1 and RW processes. The shaded blue area represents the $95\%$ confidence intervals, while the plain red, yellow and pink lines represent the contribution of the individual processes (WN, AR1 and RW respectively) to the empirical WV.}
    \label{fig:intro_plot}
\end{figure}
\begin{figure*}
\tiny
\begin{minipage}[b]{0.475\linewidth}
\centering
\includegraphics[width=\linewidth]{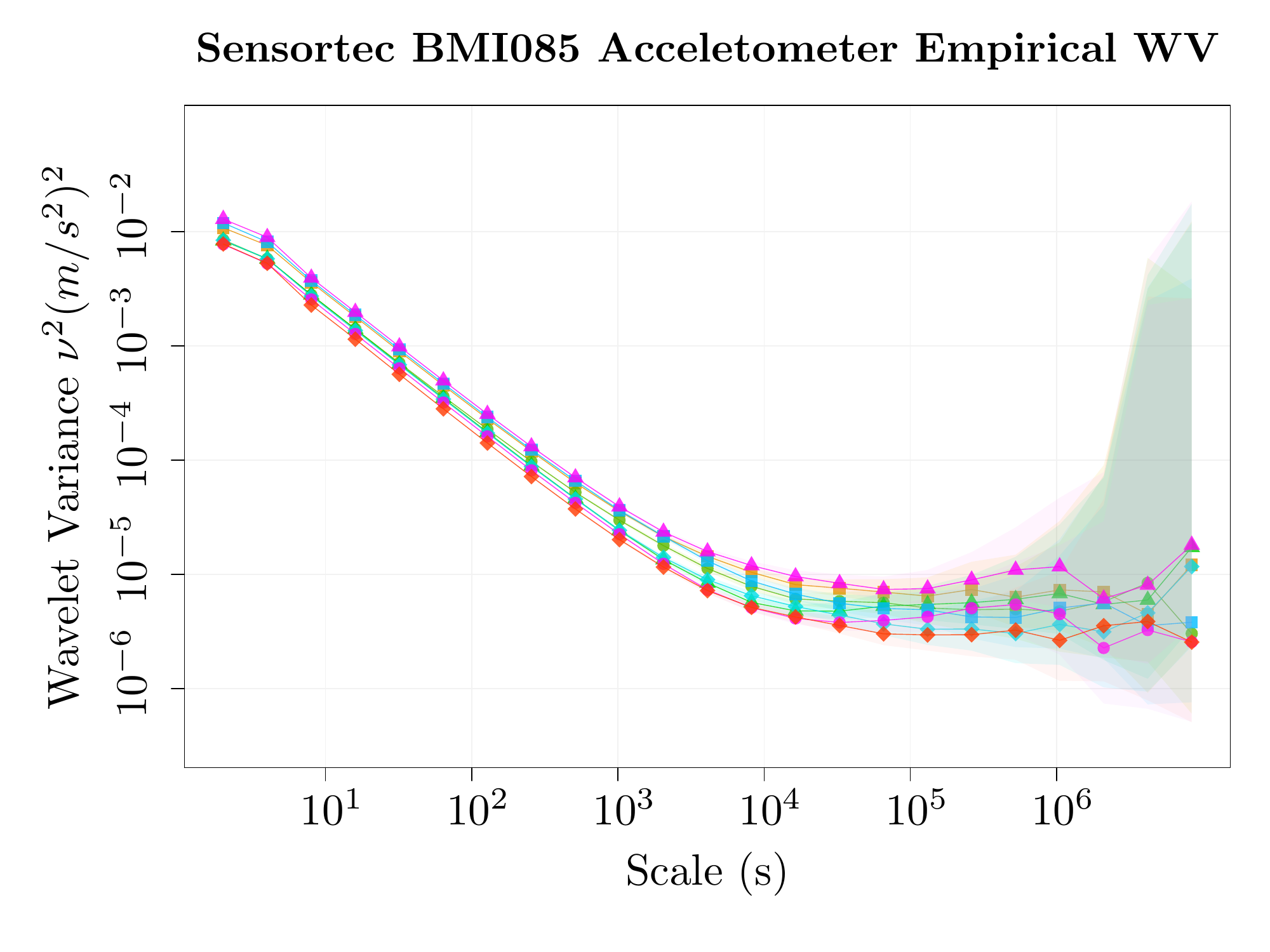}
\end{minipage}
\hspace{.3cm}
\begin{minipage}[b]{0.475\linewidth}
\centering
\includegraphics[width=\linewidth]{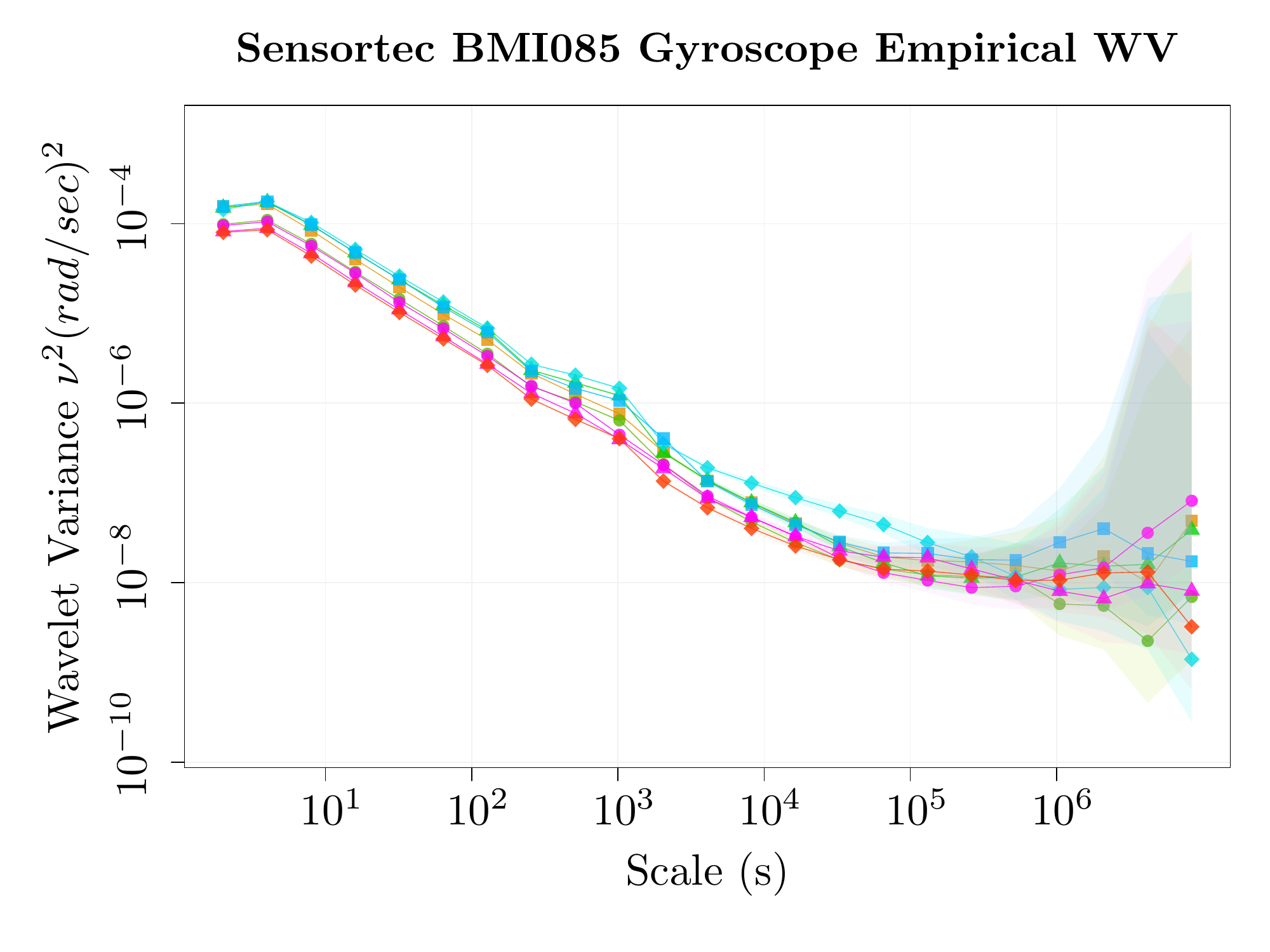}
\end{minipage}\hfill
\caption{Empirical WV (plain doted line) coming from 8 replicates of Bosch Sensortec BMI085 MEMS IMU accelerometer (left plot) and gyroscope (right plot), with their respective $95\%$ confidence intervals (shaded areas). }
\label{fig:emp_hexa}
\end{figure*}

While the GMWM has improved the task of stochastic calibration for Inertial Measurement Units (IMUs), it still relies on the common approach to calibration which consists in modelling and removing stochastic errors measured on a single experimental run (replicate) of IMU measurements in static conditions. Nevertheless, it is also common for IMU calibration procedures to perform several independent experimental runs on the same sensor from which the single signal for calibration is chosen. Although the latter approach remains a valid one, it may not be optimal since the different replicates can contain important information to appropriately model these signals for navigation purposes. Moreover, the choice of which replicate to use for calibration remains somewhat random aside from visually assessing the behavior of the single replicates thereby exposing oneself to the risk of picking a ``sub-optimal'' signal. To avoid this, the problem of considering the information from all replicates was first put forward in \cite{bakalli2017computational} and subsequently in \cite{radi2019multisignal} where it is underlined how the replicates are important to carry out a comprehensive estimation procedure but need to be used with caution due to the changes of the empirical WV between each observed signal. Fig.~\ref{fig:emp_hexa} provides an example of this behavior with eight independent recordings in static conditions coming from a Bosch Sensortec BMI085 MEMS IMU accelerometer and gyroscope. It can be seen how the shape of the WV remains roughly the same between replicates (i.e., the underlying model structure remains constant) but their values differ significantly, mainly over the first scales (i.e., their confidence intervals do not overlap). Building on this observation, \cite{radi2019multisignal} defined this setting as ``\textit{near-stationary}'' where, instead of considering a stochastic model characterized by a fixed parameter vector, they postulate that the parameters of this model are independent random variables that follow a certain stationary probability distribution $G$. Aside from requiring it to be stationary, the distribution $G$ is left unspecified by choice since, while explaining possible parameter variation due to internal sensor characteristics, it can also represent the change in parameter values due to observed and/or unobserved external factors during each calibration run. 

As a result, the near-stationary setting described above can be seen from different perspectives. From a Bayesian perspective the distribution $G$ would represent a prior distribution on the parameters \cite{lee1989bayesian} while from a random-coefficient model perspective (or mixed-model perspective) \cite{verbeke1997linear}, the parameter values of each replicate could be seen as random effects for each subject (in this case each replicate). However, in both cases, these approaches would require the exact distributional specification of $G$ in order to estimate the ``optimal'' or ``average'' parameter value. Another perspective would be a time-varying parameter model where the parameters are seen as evolving over time (between replicates), possibly due to internal or external factors. Also in this case though, a specific (parametric) model and corresponding factors would need to be specified in order to explain how these parameters evolve over time. With respect to these approaches, the setting considered in \cite{bakalli2017computational} and \cite{radi2019multisignal} (and hence in this work) takes on a semi-parametric perspective where the error signals follow a parametric model $F_{\bvartheta}$ where $\bvartheta \sim G$, with the distribution $G$ remaining unspecified. As in the Bayesian or random-coefficient framework, the goal in this setting remains to define (and estimate) the ``optimal'' parameter value that best summarizes the behavior of the random parameter $\bvartheta$ to better characterize and predict the stochastic error signals of an inertial sensor.

To address this problem, \cite{bakalli2017computational} and \cite{radi2019multisignal} put forward a solution that extends the definition of the GMWM in order to combine the information from each replicate in an adequate manner. This extension implicitly defines the parameter vector which is consequently used to best represent the overall parameter behaviour (considering that it is now a random variable) and to perform navigation updates. However, they also consider another intuitive estimator which implies that the parameter value to target is possibly different (as explained later on in this work). Given this preliminary research, this work intends to define three different solutions that can be considered for this problem (two of which are those put forward in \cite{bakalli2017computational}) and study their properties in order to clearly determine which method could be used for multi-signal calibration and under what circumstances. In addition, these properties allow to improve the understanding and performance of statistical inferential tools that can be used to assess the validity of the near-stationary hypothesis when dealing with multiple signal replicates. Moreover, these approaches remain valid for general signal/time-series problems based on moment-matching techniques and can therefore be employed in similar circumstances beyond the scope of IMU calibration.

To present and discuss the proposed approaches and results, this paper is organized as follows. Sec.~\ref{sec:ms} formally defines the near-stationary framework and discusses the properties of the different approaches considered for multi-signal calibration. These results are necessary to obtain reliable statistical estimates for navigation purposes and are confirmed in Sec.~\ref{sec:simu} which studies the finite sample performance of the three proposed approaches in a controlled simulation setting. In order to highlight the advantages of these approaches for inertial sensor calibration, Sec.~\ref{sec:case2} presents a case study on real-world inertial sensor calibration error signals which shows how the proposed approaches can generally improve the navigation performance with respect to the current setting where only one replicate is used to calibrate the inertial sensors and feed the navigation filter. Finally, Sec. \ref{sec:conclusion} concludes.

\section{Multi-Signal Calibration}
\label{sec:ms}

In this section, we present the theoretical framework of this work more formally and study the different proposed solutions for the considered setting. More specifically, in the following sections, we describe and study the solutions, including those put forward in  \cite{bakalli2017computational} and \cite{radi2019multisignal}, which are a direct extension of the GMWM. As mentioned, the latter is currently employed, among others, for sensor calibration on a \textit{single} stochastic error signal issued from an inertial sensor calibration session (see, e.g., \cite{guerrier2013wavelet,guerrier2020robust}). Indeed, in order to estimate the parameter vector ($\btheta \in \real^p$) that characterizes the model underlying the stochastic error, denoted as $F_{\btheta}$, the GMWM is defined as follows:
\begin{equation}
   \tilde{\btheta} := \argmin_{\btheta \in \bm{\Theta}} \| \hbnu - {\bnu}(\btheta)\|^2_{\bm{\Omega}},
\end{equation}
where, with $\mathbf{Z} \in \real^J$, we have that $\|\mathbf{Z}\|_{\bm{\Omega}}^2 := \mathbf{Z}^{\intercal}\bm{\Omega}\mathbf{Z}$. In addition, $\hat{\bnu} \in \real_+^J$ represents the WV estimated on the single error signal issued from the calibration session, $\bnu(\btheta) \in \real_+^J$ represents the theoretical WV implied by the parametric model $F_{\btheta}$ and $\bm{\Omega}$ is a positive definite weighting matrix chosen in a suitable way (see, e.g., \cite{guerrier2020robust} and following sections for more details).

\subsection{Near-Stationary Framework}

\noindent Compared to the setting where a single error signal is considered, a common practice for inertial sensor calibration is to independently record $K > 1$ replicates of the error signals issued from the same IMU in static conditions. Ideally, each signal (or replicate), indexed by $i$ and with length $T_i$, is issued from the same stochastic error model with the same fixed parameter values (i.e., $\btheta_i=\btheta_0, \, \forall i$) which are specific to the sensor of interest. However, based on the frequently observed random variations of plots of the WV for each signal measured on the same device (and under the same conditions), it would appear that, while the structure of the stochastic model remains the same, the parameters of the latter model appear to change between replicates. It is therefore more reasonable to assume that the parameters of the model are not fixed but vary from one signal to the next.  We therefore assume that there exists an independent sequence of random variables $\bvartheta_i$ (for $i=1,\ldots,K$), with associated probability distribution that we denote by $G$. We refer to the distribution $G$ as an \textit{internal sensor model} and we define the processes generated by the sensor as \textit{near-stationary} processes, in the sense that the model generating them (which can include non-stationary time series models) remains the same for each signal while the associated parameters change between replicates according to a probability distribution $G$ whose support is defined over a compact set $\bm{\Theta}$.

To formalize this new framework, assuming that all deterministic calibration has removed the corresponding errors (e.g. axis non-orthogonalities, etc.), let us define the $i^{th}$ stochastic error signal as $(X_t^{(i)}) \sim F_{\bvartheta_i}$ where $t = 1. \hdots, T_i$, and $\bvartheta_i \in \bm{\Theta} \subset \real^p$ which is such that $\bvartheta_i \overset{iid}\sim G$. As in the GMWM setting, the model $F_{\bvartheta_i}$ therefore represents the stochastic process governing the dependence structure over time during the $i^{th}$ calibration session, where the distribution of the innovation sequence is left unspecified (i.e., it can be Normal or another continuous distribution). With this setting in mind, we denote the estimator of WV as $\hbnu_i \in \real_+^J$ where $J$  is a \textit{fixed} integer representing the chosen number of WV scales such that $p \leq J \leq \min_i J_i$ where $J_i \in \mathbb{N}^+$ represents the number of WV scales for the $i^{th}$ signal. It must be noticed that now the estimator $\hbnu_i$ does not target a general fixed WV $\bnu(\btheta_0)$ but aims to estimate the WV implied by the random parameter vector that generated the $i^{th}$ replicate, i.e., $\bnu(\bvartheta_i)$.

Considering this new stochastic framework, it would be unreasonable to use the parameter vector estimated on the $i^{th}$  signal to predict the general measurement error of a future signal. As a consequence, it would be more appropriate to define a \textit{fixed} parameter vector that adequately represents and predicts the behaviour of all possible signals issued from the stochastic framework,  where the paramater values $\bvartheta_i$ vary from one replicate to the other. In order to do so, we adopt the parameter notation from the standard setting and define $\btheta_0$ as follows:

\begin{equation*}
    \btheta_0 := \argmin_{\btheta \in \bm{\Theta}} Q\left(\btheta\right),
    \label{eq:theta0}
\end{equation*}
where
\begin{equation}
Q\left(\btheta\right):= \mathbb{E}\left[\|\bnu(\bvartheta_i)  -  \bnu(\btheta)\|^2_{\bm{\Omega}}\right],
    \label{eq:optimality}
\end{equation}
with $\mathbb{E}[\cdot]$ denoting the expectation under the distribution $G$, $\bnu(\btheta)$ representing the theoretical WV implied by the stochastic model evaluated at the fixed parameter vector $\btheta$ and $\bm{\Omega}$ denoting a positive definite weighting matrix. With respect to the weighting matrix, for example, one can choose a \textit{fixed} positive definite matrix for $\bm{\Omega}$ (denoted as $\bm{\Omega}_0$) or, as discussed further on, an estimator of the latter matrix (denoted as $\widehat{\bm{\Omega}}$). As long as this matrix is positive definite and assuming identifiability of the function $\bnu(\cdot)$, the criterion in \eqref{eq:optimality} is always minimized in $\btheta_0$. In an estimation setting, the choice of $\bm{\Omega}$ is usually limited to minimizing the asymptotic variance which is achieved by choosing $\bm{\Omega} := \bm{V}^{-1}$, where $\bm{V}$ is the asymptotic covariance matrix of the estimated WV (see \cite{guerrier2013wavelet, guerrier2020robust}), although a simple diagonal matrix (such as the identity) can often be more than sufficient in practice. With this in mind, the criterion (or loss/objective function) in \eqref{eq:optimality} is an extension of the GMWM objective function which takes into account the internal sensor model $G$. The logic behind choosing this criterion therefore consists in finding a fixed parameter vector $\btheta_0$ that minimizes the expected squared-loss between the WV implied by the latter parameter and the WV implied by all possible values of the (parameter) random variable $\bvartheta_i$. In a Bayesian sense, we are finding the optimal parameter value (according to the GMWM criterion) weighted by the prior distribution $G$ which however does not need to be specified since this expectation is evaluated empirically through observed ``realizations'' or ``representations'' of the distribution $G$ as presented in the following paragraphs.

\subsection{Multi-Signal Approaches}

Given that we cannot directly observe the criterion in \eqref{eq:optimality} that would allow us to find the value of $\btheta_0$, we need to consider estimators for this quantity. For this reason, this work studies different solutions, among which those put forward in \cite{bakalli2017computational} and \cite{radi2019multisignal} whose finite sample performance was investigated through preliminary simulations and applied studies. These solutions are intuitive estimators for the quantity of interest $\btheta_0$ but, as shown further on, have different properties and actually turn out to be the same under specific or more general circumstances. 

However, compared to the solutions put forward in \cite{bakalli2017computational} and \cite{radi2019multisignal}, we define a more general setting where we can assign weights to the information coming from each replicate. More specifically, we define the weights that characterize the studied solutions as follows:
\begin{equation*}
    w_i := d_i\,\frac{T_i}{\sum_{j=1}^K T_j},
    \label{eq:weigths}
\end{equation*}
where $d_i$ is a signal-specific constant defined by the user to give more weight to certain signals based on prior knowledge (one would however commonly choose $d_i = 1$ for all $i$). Based on this definition, conditioned on the choice of $d_i$, these weights are larger for longer signals therefore giving more weight to those signals that carry more information. It must be noticed that, considering the case where $d_i=1\,\,\forall i$, these weights have the following properties
\begin{equation}
    \sum^K_{i = 1} w_i = 1, \,\, w_i \geq 0 \,\, \text{and} \,\, w_i = \mathcal{O}(K^{-1}),
    \label{eq:weigths2}
\end{equation}
which are important to determine the theoretical properties of the estimators studied in the following paragraphs. If $d_i\neq 1$, then they should be chosen such that  \eqref{eq:weigths2} holds.

Considering the multiple signal recording setting formalized in the previous paragraphs, the goal of the methods studied in this work is to combine the information from the different signals in an optimal (weighted) manner. The first and most intuitive way to do so would be to take a simple weighted average of the GMWM estimators issued from the individual signals (we refer to this estimator as the Average GMWM (AGMWM) which was suggested in \cite{bakalli2017computational}). More formally, this estimator is defined as follows:
\begin{equation}
    \hbtheta^{\circ} := \sum^K_{i = 1} w_i \tilde{\bvartheta}_i,
    \label{eq:averagetheta}
\end{equation}
where 
$$\tilde{\bvartheta}_i := \argmin_{\bvartheta_i \in \bm{\Theta}} \| \hbnu_i - {\bnu}(\bvartheta_i)\|^2_{\bm{\Omega}},$$
are the individual GMWM parameter estimates for each signal. This estimator can also be defined as follows:
\begin{equation}
    \hbtheta^{\circ} = \argmin_{\btheta \in \bm{\Theta}} \widehat{Q}^{\circ}(\btheta),
    \label{eq:averagetheta_min}
\end{equation}
where 
$$\widehat{Q}^{\circ}(\btheta) := \| \sum^K_{i = 1} w_i \tilde{\bvartheta}_i - \btheta\|_{\bm{I}}^2,$$
(with $\bm{I}$ being the identity matrix). Based on this definition, it is clear that the criterion defining the AGMWM does not correspond to the objective function in \eqref{eq:optimality}.

The second estimator is new and we call it the Average WV (AWV) estimator which is defined as follows:
\begin{equation}
    \hbtheta^{\dagger} := \argmin_{\btheta \in \bm{\Theta}} \widehat{Q}^{\dagger}(\btheta),
    \label{eq:averagewv}
\end{equation}
where $$\widehat{Q}^{\dagger}(\btheta):= \| \sum^K_{i = 1} w_i \hat{\bnu}_i - \bnu(\btheta))\|^2_{\bm{\Omega}}.$$ The idea behind this estimator is to replicate the structure of the GMWM estimator and, instead of considering a single estimate of the WV, we take the weighted average of the individual estimated WV. The objective function defining this estimator also resembles the criterion given in \eqref{eq:optimality} and, as we will see further on, indeed targets this criterion.

The final estimator we study is the weighted version of the estimator defined in \cite{bakalli2017computational} and \cite{radi2019multisignal} and is given by the solution to the objective function resulting from the weighted average of the individual GMWM objective functions. More specifically, this estimator, referred to as the Multi-Signal GMWM (MS-GMWM), is defined as
\begin{equation}
    \hbtheta := \argmin_{\btheta \in \bm{\Theta}} \widehat{Q}(\btheta), 
	\label{eq:msgmwm}
\end{equation}
where
$$\widehat{Q}(\btheta):=\sum^K_{i=1} w_i \| \hbnu_i - \bnu(\btheta)\|^2_{\bm{\Omega}}\,.$$
This estimator is therefore the result of the minimization of a direct estimator of the criterion in \eqref{eq:optimality}. Indeed, the empirical WV $\hbnu_i$ is an estimator for the theoretical quantity $\bnu (\bvartheta_i)$ while the weighted sum over the $K$ signals is aimed at estimating the theoretical expectation $\mathbb{E}[\cdot]$ under the internal sensor model $G$.

\subsection{Statistical Properties}
\noindent Having formally defined the methods of interest for the problem at hand, we now lay out a series of assumptions that are necessary to define the asymptotic properties of these estimators. For this reason, we also define $\mathbf{H}(\btheta):= \bm{A}(\btheta)^\intercal\bm{\Omega}\bm{A}(\btheta)$ where

$$\bm{A}(\btheta) := \frac{\partial}{\partial \bvartheta^{\top}}\,\bm{\nu}(\bvartheta)\Big|_{\bvartheta = \btheta}\,.$$
\vspace{0.2cm}

\setcounter{Assumption}{0}
\renewcommand{\theHAssumption}{otherAssumption\theAssumption}
\renewcommand\theAssumption{\Alph{Assumption}}

\begin{Assumption}[Parameter Space]
\label{assum:compactness}
$\btheta_0$ is an interior point of the set $\bm{\Theta}$ which is compact.
\end{Assumption}\vspace{0.25cm}
\vspace{0.4cm}

\begin{Assumption}[Theoretical WV]
\label{assum:injectivity}
The theoretical WV is such that:
\begin{itemize}
    \item $\bnu(\btheta)$ is continuously differentiable $\forall \btheta \in \Theta$;
    \item $\bnu(\btheta_1) = \bnu(\btheta_2)$ if and only if $\btheta_1 = \btheta_2$;
    \item $\mathbf{H}(\btheta_0)$ exists and is non-singular.
\end{itemize}
\end{Assumption}
\vspace{0.4cm}


\begin{Assumption}[Asymptotics]
\label{assum:consistent}
\begin{equation*}
        \|\hbnu_i - \bnu(\bvartheta_i)\| = o_p(1).
    \end{equation*}
Moreover, if $\bm{\Omega}_0$ is estimated by $\widehat{\bm{\Omega}}$, then we have that 
\begin{align*}
    ||\widehat{\bm{\Omega}} - \bm{\Omega}_0||_{S} = o_{\rm p}(1).
\end{align*}
where $|| \cdot ||_S$ denotes the matrix spectral norm.
\end{Assumption}
\vspace{0.4cm}

Assumption \ref{assum:compactness} is a standard regularity condition that ensures that certain quantities are bounded and to allow convergence (however it can be partly relaxed depending on the model of interest). Assumption \ref{assum:injectivity} ensures that (i) $\bnu(\cdot)$ is differentiable (in order to perform expansions); (ii) $\bnu(\cdot)$ is injective (in order to have identifiability); and (iii) certain quantities from these expansions exist in order to prove consistency and asymptotic normality of the estimators. Finally, Assumption \ref{assum:consistent} requires consistency of the WV estimator (which was proven under different conditions, see e.g. \cite{guerrier2020robust}, also not requiring Gaussianity of the processes, see e.g. \cite{xu2019multivariate}) as well as that of $\widehat{\bm{\Omega}}$ (if an estimator is actually chosen for the weighting matrix $\bm{\Omega}$). 

These assumptions are required to prove results on consistency and asymptotic normality of the multi-signal approaches described earlier which can provide insight to convergence rates of these approaches as well as justify the use of time-dependent bootstrap methods to deliver adequate uncertainty quantification for each of them. Denoting $T := \underset{i}{\min}\, T_i$, we can now study the first of the considered estimators, namely the AGMWM. For this estimator, we consider the quantity $\btheta^{\circ} := \mathbb{E}\left[\bvartheta_i\right]$  and define
$$\bm{\Lambda}^\circ := \mathbb{E}\left[\bm{H}(\bvartheta_i)^{-1}\bm{A}(\bvartheta_i)^{\top} \bm{\Omega}\bm{V}_i\bm{\Omega}\bm{A}(\bvartheta_i)\bm{H}(\bvartheta_i)^{-1}\right],$$
where $\bm{V}_i := \mathbb{V}[\hbnu_i]$.

\begin{Theorem}
\label{prop.agmwm}
Under Assumptions \ref{assum:compactness} to \ref{assum:consistent} and letting $K,T~\to~\infty$, we have that
$$\sqrt{KT}(\hbtheta^\circ - \btheta^{\circ}) \xrightarrow{\mathcal{D}} \mathcal{N}\left(0, \bm{\Lambda}^\circ\right).$$
\end{Theorem}

\vspace{0.4cm}

\begin{proof}
The proof of Theorem \ref{prop.agmwm} is straightforward since the individual GMWM estimators are consistent for the respective parameters $\bvartheta_i$ under Assumptions \ref{assum:compactness} to \ref{assum:consistent}  (see e.g. \cite{guerrier2020robust}) and, using Theorem 1 of \cite{jamison1965convergence} and Assumption \ref{assum:compactness}, the weighted average of the GMWM estimators $\tilde{\bvartheta}_i$ will converge to their expectation $\mathbb{E}[\bvartheta_i]$ (i.e. $\btheta^{\circ}$). Based on this, we have that the weights respect the conditions in Theorem 1 of \cite{weber2006weighted} and again using Assumption \ref{assum:compactness} we have that $\sqrt{KT}(\hbtheta^\circ - \btheta^{\circ})$ tends to a normal distribution thus concluding the proof.
\end{proof}
\vspace{0.2cm}

From this result, it can be noticed how the AGMWM targets the expected value of the internal sensor model G which does not necessarily correspond to the desired value $\btheta_0$ defined in \eqref{eq:optimality}, except in specific circumstances stated further on. 

\begin{Remark}
\label{rem.unbiased}
When using the unbiased Maximal
Overlap Discrete Wavelet Transform (MODWT) estimator for $\hbnu_i$ (see e.g. \cite{serroukh2000statistical}), then the result of Theorem \ref{prop.agmwm} holds without letting $T \to \infty$ (hence the normalizing factor would only consist in $\sqrt{K}$) and the covariance matrix of $\hbtheta^\circ$ could be denoted as $\bm{\Gamma}(T)$ underlining its dependence on the minimum signal size $T$. This remark holds also for the results on the other estimators studied in the following paragraphs.
\end{Remark}

\vspace{0.2cm}

Considering that the AGMWM does not necessarily target the quantity of interest $\btheta_0$, we now study the AWV estimator whose objective function appears closer to the form of the criterion in \eqref{eq:optimality}. Indeed, the AWV estimator targets the desired quantity as stated in the following theorem.

\begin{Theorem}
Under Assumptions \ref{assum:compactness} to \ref{assum:consistent} and letting $K,T~\to~\infty$, we have that

\begin{equation*}
        \|\hbtheta^{\dagger} - \btheta_0\| = o_p(1).
    \end{equation*}
    \label{lemma:msgmwm}
\end{Theorem}
\vspace{0.4cm}

\begin{proof}
The expectation of a quadratic form such as that in $Q(\btheta)$ can be written as:
$$Q(\btheta) = \underbrace{\|\mathbb{E}[\bnu(\bvartheta_i) - \bnu(\btheta)]\|_{\bm{\Omega}}^2}_{\widetilde{Q}(\btheta)} + \tr(\bm{\Omega}\mathbb{V}[\hbnu_i]),$$
where $\tr(\cdot)$ indicates the matrix trace and $\mathbb{V}[\hbnu_i]$ denotes the variance of the WV estimator. Since the second term does not depend on the fixed parameter $\btheta$, the criteria $Q(\btheta)$ and $\widetilde{Q}(\btheta)$ are both minimized in the same point (i.e., $\btheta_0$) based on Assumption \ref{assum:injectivity}. Given Assumptions \ref{assum:compactness} to \ref{assum:consistent} (see, e.g., \cite{guerrier2020robust}), we therefore only need to prove 
$$\underset{\bm{\theta} \in \bm{\Theta}}{\sup}\big|\widehat{Q}^{\dagger}(\btheta) - \widetilde{Q}(\btheta)\big| \xrightarrow{P} 0.$$
To simplify notation, we use $\bar{\bnu}:=\sum_{i=1}^K w_i \hbnu_i$. Firstly, based on Assumption \ref{assum:consistent} we have that $\hbnu_i$ converges to $\bnu(\bvartheta_i)$ which is bounded based on Assumptions \ref{assum:compactness} and \ref{assum:injectivity}. Hence, using Theorem 1 of \cite{jamison1965convergence}, Assumption \ref{assum:compactness} and the continuity of the norm, we have that
\begin{eqnarray*}
\underset{K \to \infty}{\lim} \widehat{Q}^{\dagger}(\btheta) &=& \| \underset{K \to \infty}{\lim} \bar{\bnu} - \bnu(\btheta)\|^2_{\bm{\Omega}}\\
&=& \| \mathbb{E}[\hbnu_i] - \bnu(\btheta)\|^2_{\bm{\Omega}}.
\end{eqnarray*}
The expression in the norm can be written as
$$\mathbb{E}[\hbnu_i] - \bnu(\btheta) = \mathbb{E}[\mathbb{E}_X\left[\hbnu|\bvartheta_i\right]] - \bnu(\btheta),$$
where $\mathbb{E}_X\left[\hbnu|\bvartheta_i\right]$ is the conditional expectation of the WV estimator $\hbnu$ given  $\bvartheta_i$. Based on Assumption \ref{assum:consistent}, we can write this last expression as
$$\mathbb{E}[\bnu(\bvartheta_i) + o_p(1)] - \bnu(\btheta) = \mathbb{E}[\bnu(\bvartheta_i) - \bnu(\btheta)] + o_p(1),$$
which allows us to express the norm as
$$\|\mathbb{E}[\bnu(\bvartheta_i) - \bnu(\btheta)]\|_{\bm{\Omega}}^2 + o_p(1).$$
Based on the latter, we prove that 
$$\underset{\bm{\theta} \in \bm{\Theta}}{\sup}\big|\widehat{Q}^{\dagger}(\btheta) - \widetilde{Q}(\btheta)\big| \xrightarrow{P} 0,$$
which, using Theorem 2.1 of \cite{newey1994large} with Assumptions \ref{assum:compactness} to \ref{assum:consistent} concludes the proof.
\end{proof}
\vspace{0.2cm}

Theorem~\ref{lemma:msgmwm} therefore shows that the AWV targets the desired quantity and is therefore preferable over the AGMWM if one aims at minimizing the criterion in \eqref{eq:optimality}. In addition, as underlined in Remark \ref{rem.unbiased}, if using the unbiased MODWT estimator Theorem~\ref{lemma:msgmwm} result would hold also in the case where $T$ does not diverge. 

The third estimator that we would need to study is the MS-GMWM. However, the following proposition underlines how the two estimators (AWV and MS-GMWM) are actually the same estimator.

\begin{Proposition}
\label{prop.equality}
Under Assumptions \ref{assum:compactness} and \ref{assum:injectivity}, we have that
$$\hbtheta^{\dagger} = \hbtheta\,.$$
\end{Proposition}

\vspace{0.4cm}

\begin{proof}
Under Assumptions \ref{assum:compactness} and \ref{assum:injectivity}, the AWV and MS-GMWM estimators can be defined in terms of their derivatives, i.e.,
$$\hbtheta^{\dagger} := \argzero_{\btheta \in \bm{\Theta}} \frac{\partial}{\partial \btheta} \widehat{Q}^{\dagger}(\btheta),$$
and
$$\hbtheta := \argzero_{\btheta \in \bm{\Theta}} \frac{\partial}{\partial \btheta} \widehat{Q}(\btheta),$$
respectively, where $\argzero$ stands for the value of $\btheta$ that allows the expression to be zero. Considering this, the derivative of $\widehat{Q}(\btheta)$ is given by
\begin{eqnarray*}
\frac{\partial}{\partial \btheta} \widehat{Q}(\btheta) &=& \frac{\partial}{\partial \btheta} \sum_{i=1}^K w_i \|\hbnu_i - \bnu(\btheta)\|_{\bm{\Omega}}^2\\
&=& -\sum_{i = 1}^K 2 w_i \frac{\partial}{\partial \btheta} \bnu(\btheta) \bm{\Omega}(\hbnu_i - \bnu(\btheta))\\
&=& -2 \frac{\partial}{\partial \btheta} \bnu(\btheta) \bm{\Omega} \sum_{i = 1}^K w_i(\hbnu_i - \bnu(\btheta)).
\end{eqnarray*}
Knowing that $\sum_{i = 1}^K w_i = 1$, we finally have that
$$\frac{\partial}{\partial \btheta} \widehat{Q}(\btheta) = -2 \frac{\partial}{\partial \btheta} \bnu(\btheta) \bm{\Omega} \left(\sum_{i = 1}^K w_i\hbnu_i - \bnu(\btheta)\right)\,.$$
If we take the derivative of $\widehat{Q}^{\dagger}(\btheta)$, we obtain
\begin{eqnarray*}
\frac{\partial}{\partial \btheta} \widehat{Q}^{\dagger}(\btheta) &=& \frac{\partial}{\partial \btheta} \|\sum_{i=1}^K w_i\hbnu_i - \bnu(\btheta)\|_{\bm{\Omega}}^2\\
&=& -2 \frac{\partial}{\partial \btheta} \bnu(\btheta) \bm{\Omega}\left(\sum_{i = 1}^K w_i \hbnu_i - \bnu(\btheta)\right).\\
\end{eqnarray*}
Since $\widehat{Q}(\btheta)$ and $\widehat{Q}^{\dagger}(\btheta)$ have the same derivative, under Assumptions \ref{assum:compactness} and \ref{assum:injectivity} they have the same solution in zero and, consequently, we have that $\hbtheta^{\dagger} = \hbtheta$ thus concluding the proof.
\end{proof}
\vspace{0.2cm}

Given Proposition \ref{prop.equality}, we do not need to study the properties of the MS-GMWM since they will be the same as those of the AWV. Considering this, having proved consistency of the AWV, let us now deliver the final property of the AWV which consists in its asymptotic distribution.

\begin{Proposition}
Under Assumptions \ref{assum:compactness} to \ref{assum:consistent} and letting $K,T~\to~\infty$, we have that
$$\sqrt{KT}(\hbtheta^{\dagger} - \btheta_0) \xrightarrow{\mathcal{D}} \mathcal{N}\left(0, \bm{\Lambda}_0\right),$$
where $\bm{\Lambda}_0
            := \bm{H}(\bm{\theta}_0)^{-1}\bm{A}(\bm{\theta}_0)^{\intercal} \bm{\Omega}\bar{\bm{V}}\bm{\Omega}\bm{A}(\bm{\theta}_0)\bm{H}(\bm{\theta}_0)^{-1}$ is the asymptotic covariance matrix, with $\bar{\bm{V}}:= \mathbb{E}[\bm{V}_i]$.
    \label{prop:msgmwm}
\end{Proposition}

\vspace{0.4cm}

\begin{proof}
This proof is adapted and closely follows the proof of Lemma 3.1 in \cite{guerrier2020robust}. More specifically, given the results on the consistency in Theorem \ref{lemma:msgmwm}, the proof of asymptotic normality of $\hbtheta^{\dagger}$ naturally follows the standard proof of asymptotic normality for extremum estimators (see e.g. \cite{newey1994large}). Indeed, using again the notation $\bar{\bnu} := \sum_{i=1}^K w_i \hbnu_i$ and under Assumption \ref{assum:injectivity}, by the definition of $\hbtheta^{\dagger}$, we have 
    \begin{equation*}
        \begin{aligned}
            &\frac{\partial \widehat{Q}^{\dagger}(\btheta)}{\partial \btheta} \bigg\rvert_{\btheta = \hbtheta^{\dagger}} = \bm{0}_{p \times 1} \\
            \Longleftrightarrow& \frac{\partial}{\partial \bm{\theta}}\left[\left( \bar{\bnu} - \bnu(\btheta) \right)^{\intercal} \bm{\Omega} \left( \bar{\bnu} - \bnu(\btheta) \right) \right] \Bigg\rvert_{\btheta = \hbtheta^{\dagger}} = \bm{0}_{p \times 1}
        \end{aligned}
    \end{equation*}
    which, up to a constant, yields
    \begin{equation}
        \label{eq:1order}
        \underbrace{\left(\frac{\partial}{\partial \bm{\theta}} \left( \bar{\bnu} - \bm{\nu}(\bm{\theta}) \right)^{\intercal}\big\rvert_{\bm{\theta} = \hbtheta^{\dagger}}\right)}_{\bm{B}(\hat{\bm{\theta}}{\color{blue}^{\dagger}})} \bm{\Omega} \left( \bar{\bnu} - \bm{\nu}(\hbtheta^{\dagger}) \right)  = \bm{0}_{p \times 1}.
    \end{equation}
    The multivariate mean value theorem ensures that, based on Assumption \ref{assum:compactness}, there exists a matrix $A(\hbtheta^{\dagger}, \btheta_0)$ that can be used to expand $\bar{\bnu} - \bm{\nu}(\hbtheta^{\dagger})$ around $\btheta_0$ in the following way
    \begin{equation}
        \label{eq:maclaurin}
        \bar{\bnu} - \bm{\nu}(\hbtheta^{\dagger}) = \bar{\bnu} - \bm{\nu}(\bm{\theta}_0) + 
        \bm{A}(\hbtheta^{\dagger}, \btheta_0)
         \left( \hbtheta^{\dagger} - \btheta_0 \right).
    \end{equation}
    Based on the derivatives in the proofs of Proposition \ref{prop.equality} (whose solutions for zero occur when $\bar{\bnu} = \bnu(\btheta)$) and using Theorem \ref{lemma:msgmwm}, we have that  $\bar{\bnu} \overset{p}{\to} \bm{\nu}(\bm{\theta}_0)$ and $\hbtheta^{\dagger} \overset{p}{\to} \btheta_0$ such that the multivariate mean value theorem also guarantees that the matrix $\bm{A}(\hbtheta^{\dagger}, \btheta_0)$ has the following property
    \begin{equation*}
        \bm{A}(\hbtheta^{\dagger}, \btheta_0) \overset{p}{\to} \frac{\partial}{\partial \bm{\theta}^{\intercal}} \left( \bm{\nu}(\bm{\theta}_0) - \bm{\nu}(\bm{\theta})\right)\bigg\rvert_{\bm{\theta} = \bm{\theta}_0} = \frac{\partial}{\partial \bm{\theta}^{\intercal}} \bm{\nu}(\bm{\theta})\bigg\rvert_{\bm{\theta} = \bm{\theta}_0},
    \end{equation*}
    given that $\nicefrac{\partial}{\partial \bm{\theta}}\,\bm{\nu}(\bm{\theta}^{\intercal})$ is continuous. Plugging (\ref{eq:maclaurin}) in the third factor of (\ref{eq:1order}), multiplying by $\sqrt{KT}$ and using Assumption \ref{assum:injectivity} allows us to state that $\sqrt{KT}\left( \hbtheta^{\dagger} - \btheta_0 \right)$ is equal to

    \begin{equation}
            -\left[\bm{B}(\hbtheta^{\dagger})\,\bm{\Omega} \,\bm{A}(\hbtheta^{\dagger}, \btheta_0)\right]^{-1} \bm{B}(\hbtheta^{\dagger})\,\bm{\Omega} \,\sqrt{KT} \left( \bar{\bnu} - \bm{\nu}(\btheta_0) \right).
            \label{eq.asy_form}
    \end{equation}
Knowing that $\bm{B}(\hbtheta^{\dagger}) \overset{p}{\to} \bm{A}(\bm{\theta}_0)^{\intercal}$ by the continuous mapping theorem, and that $\hbtheta^{\dagger} \overset{p}{\to} \btheta_0$ from Theorem~\ref{lemma:msgmwm}, by Slutsky's theorem we have that
\begin{equation*}
            \left[\bm{B}(\hbtheta^{\dagger})\,\bm{\Omega} \,\bm{A}(\hbtheta^{\dagger}, \bm{\theta}_0)\right]^{-1} \bm{B}(\hbtheta^{\dagger})\,\bm{\Omega}
    \end{equation*}
    converges in probability to
\begin{equation*}
    \left[ \bm{A}(\bm{\theta}_0)^{\intercal}\bm{\Omega} \bm{A}(\bm{\theta}_0)\right]^{-1} \bm{A}(\bm{\theta}_0)^{\intercal} \bm{\Omega}.
\end{equation*}
By again using Slutsky's theorem as well Theorem 1 of \cite{weber2006weighted} in conjunction with Assumption \ref{assum:compactness}, we have that (\ref{eq.asy_form}) has the following asymptotic distribution
    \begin{equation*}
    \sqrt{KT}\left( \hbtheta^{\dagger} - \bm{\theta}_0 \right) \xrightarrow[KT\rightarrow\infty]{\mathcal{D}} \mathcal{N}(\bm{0},\bm{\Lambda}_0),
    \end{equation*}
     where $\bm{\Lambda}_0$ is given by
     \begin{equation*}
            \bm{\Lambda}_0
            := \bm{H}(\bm{\theta}_0)^{-1}\bm{A}(\bm{\theta}_0)^{\intercal} \bm{\Omega}\bar{\bm{V}}\bm{\Omega}\bm{A}(\bm{\theta}_0)\bm{H}(\bm{\theta}_0)^{-1},
    \end{equation*}
    and $\bar{\bm{V}} := \mathbb{E}[\bm{V}_i]$ thus concluding the proof.
\end{proof}
\vspace{0.2cm}

As a consequence of Proposition~\ref{prop.equality} and \ref{prop:msgmwm}, we can also state the following corollary.

\begin{Corollary}
Under Assumptions \ref{assum:compactness} to \ref{assum:consistent} and letting $K,T~\to~\infty$, we have that
$$\sqrt{KT}(\hbtheta - \btheta_0) \xrightarrow{\mathcal{D}} \mathcal{N}\left(0, \bm{\Lambda}\right).$$
    \label{corol:msgmwm}
\end{Corollary}

We omit the proof of this corollary since it is a direct consequence of Proposition \ref{prop.equality}. Again, the above results on consistency and asymptotic normality of the AWV (and consequently MS-GMWM) would hold without letting $T$ diverge if one employs the unbiased MODWT estimator (as stated for example in Remark \ref{rem.unbiased}). 

We conclude this section by delivering one final result which states the case under which the AGMWM actually targets the desired quantity $\btheta_0$. This result is provided in the following proposition where $\bm{W}$ denotes a non-singular matrix.

\begin{Proposition}
If the theoretical WV is such that $\bnu(\btheta) = \bm{W}\btheta$ we have that
$$\hbtheta^{\circ} = \hbtheta^{\dagger} = \hbtheta\,.$$
\end{Proposition}

\begin{proof}
When $\bnu(\btheta) = \bm{W}\btheta$, as shown in \cite{guerrier2020wavelet}, the GMWM has an explicit solution given by
$$\bar{\bvartheta} = (\bm{W}^{\intercal}\bm{\Omega}\bm{W})^{-1}\bm{W}^{\intercal}\bm{\Omega}\hbnu.$$
\begin{figure*}[!ht]
\begin{minipage}[b]{0.475\linewidth}
\centering
\includegraphics[width=\linewidth]{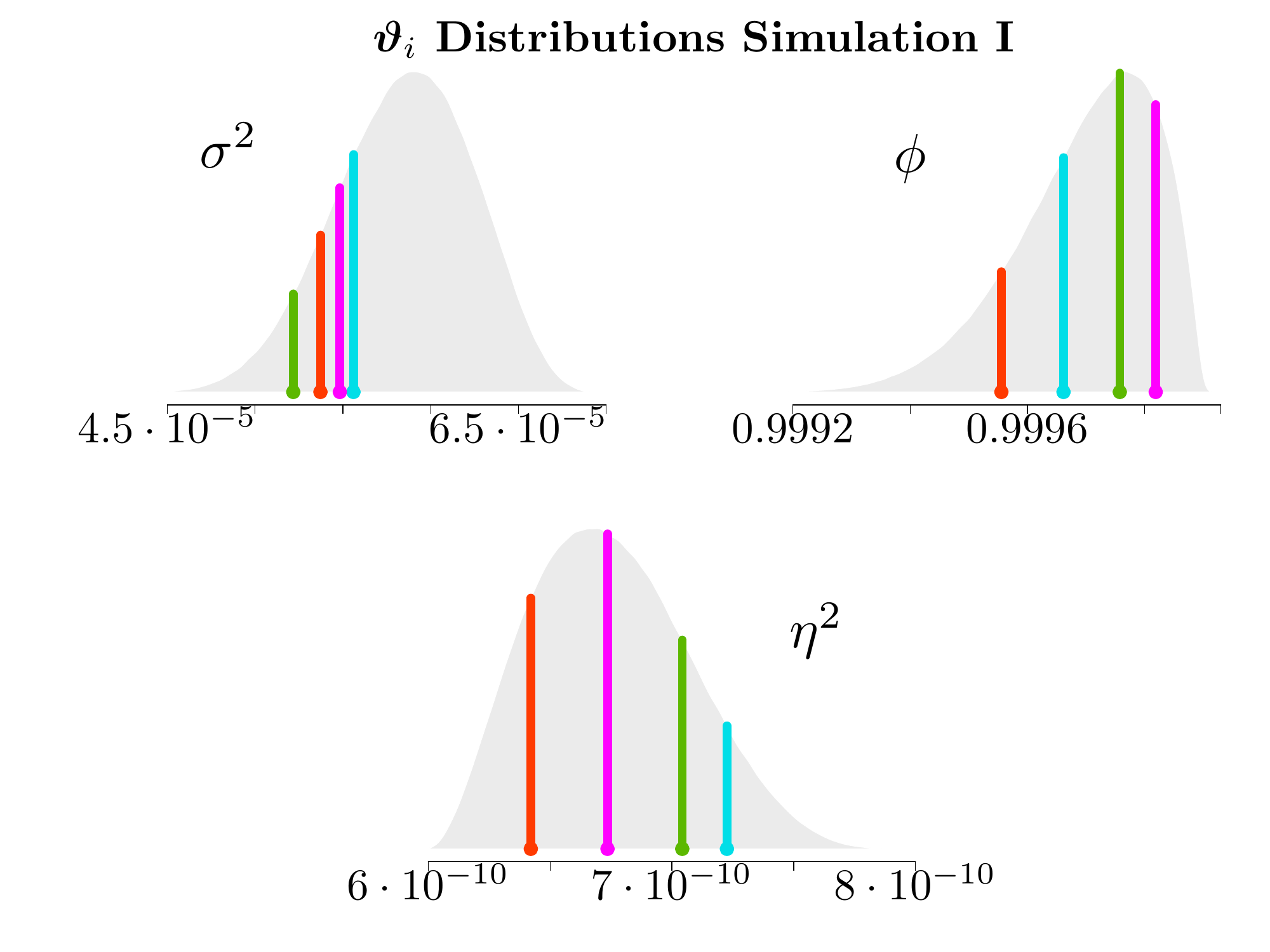}
\end{minipage}\hfill
\begin{minipage}[b]{0.475\linewidth}
\centering
\includegraphics[width=\linewidth]{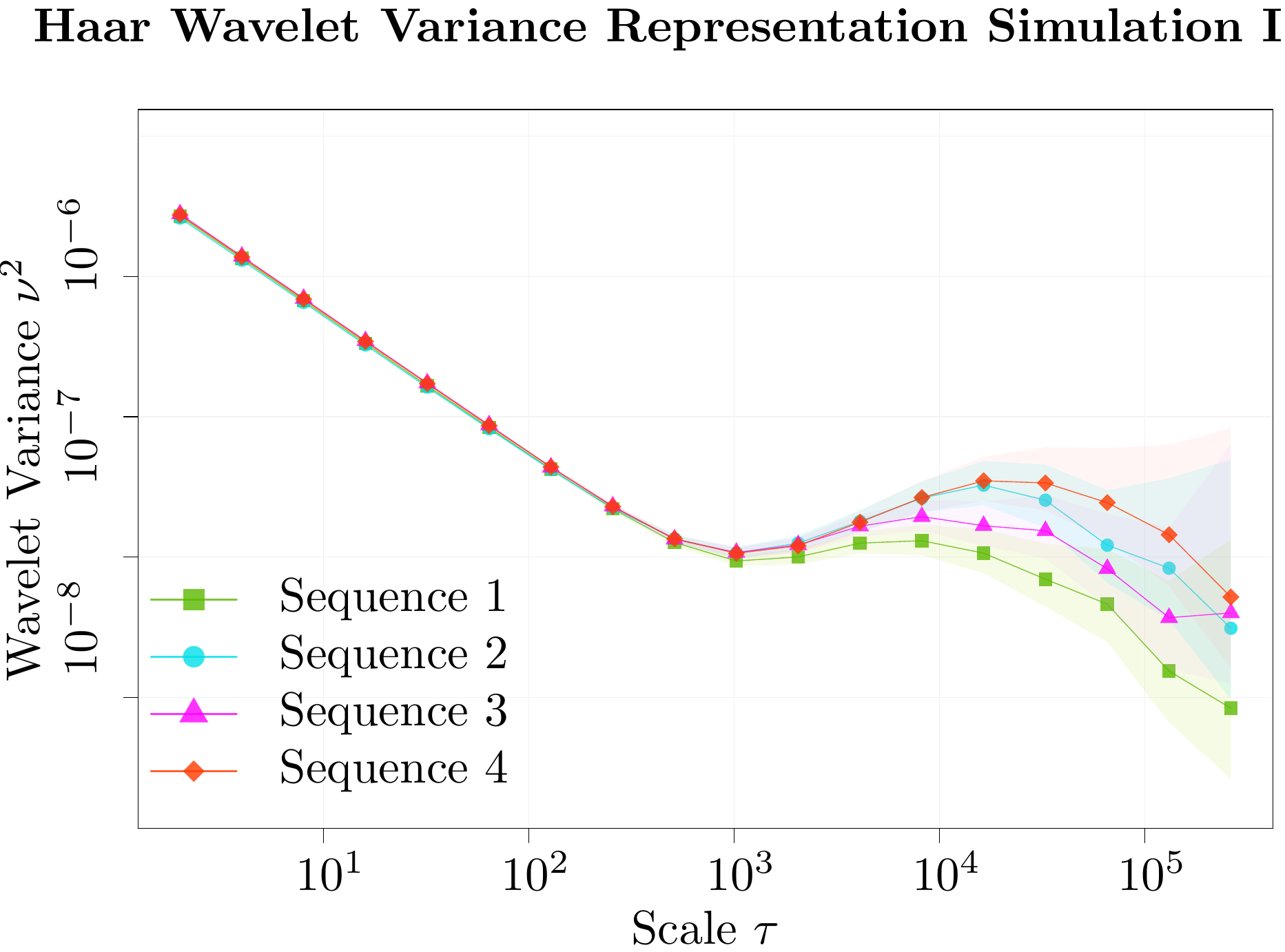}
\end{minipage}\hfill
\caption{Left: Marginal densities of the internal sensor model $G$ for the parameters $\sigma_i^2$, $\phi_i$ and $\eta_i^2$ considered in the Simulation I setting (WN + AR1) with horizontal colored lines representing four randomly selected values from each density. Right: WV plots and 95\% confidence intervals for the empirical WV of the signals generated by the parameter values selected from the respective densities in the top part (each color in the top part corresponds to the color of the WV in the bottom part).}
\label{fig_wv_simu1}
\end{figure*}
Hence, in this case the AGMWM estimator can be expressed as 
\begin{eqnarray*}
\hbtheta^{\circ} &=& \sum_{i=1}^K w_i \tilde{\bvartheta}_i\\
&=& \sum_{i=1}^K w_i (\bm{W}^{\intercal}\bm{\Omega}\bm{W})^{-1}\bm{W}^{\intercal}\bm{\Omega}\hbnu_i\\
&=& (\bm{W}^{\intercal}\bm{\Omega}\bm{W})^{-1}\bm{W}^{\intercal}\bm{\Omega}\underbrace{\sum_{i=1}^K w_i \hbnu_i}_{\bar{\bnu}}.
\end{eqnarray*}
Now, based on the proof in Proposition \ref{prop.equality}, we have that $\hbtheta^{\dagger}$ is the solution in $\btheta$ of the following equation
$$2 \frac{\partial}{\partial \btheta} \bnu(\btheta) \bm{\Omega}\left(\underbrace{\sum_{i = 1}^K w_i \hbnu_i}_{\bar{\bnu}} - \bnu(\btheta)\right) = 0,$$
which, in the case where $\bnu(\btheta) = \bm{W}(\btheta)$, delivers
$$2\bm{W}^{\intercal}\bm{\Omega}\left(\bar{\bnu} - \bm{W}\btheta\right) = 0.$$
The solution is therefore given by 
$$\hbtheta^{\dagger} = (\bm{W}^{\intercal}\bm{\Omega}\bm{W})^{-1}\bm{W}^{\intercal}\bm{\Omega}\bar{\bnu},$$
which is the same as for the AGMWM and, based on Proposition \ref{prop.equality}, the same as for the MS-GMWM.
\end{proof}
\vspace{0.2cm}

This last result therefore states that, whenever the process underlying the signals delivers a theoretical WV which is linear in the parameters of interest, the parameter $\btheta_0$ can be estimated with any of the three solutions considered in this work, including the AGMWM. Examples of such processes are the white noise, quantization noise, random walk and drift, or a combination thereof.

\subsection{Discussion}

The three solutions considered in this work therefore all have appropriate asymptotic properties under the stated assumptions. However, these results show that the AGMWM, considered in \cite{bakalli2017computational}, is not generally adequate if one intends to target the quantity $\btheta_0$ unless the processes underlying the signals have a linear WV, which may not always be the case since often the signals are characterized by autoregressive (or Gauss-Markov) processes whose WV are not linear in the parameters. Indeed, the AGMWM targets the expected value of the internal sensor model which may not be the optimal quantity to use within a navigation filter for prediction purposes. On the other hand, the other two estimators (AWV and MS-GMWM) have been proven to be the same and can therefore be used interchangeably to estimate the parameter of interest $\btheta_0$. The only arguments in favor of choosing one of the latter estimators over the other are practical in nature. More precisely, the AWV has a practical advantage from an implementation perspective since it can directly rely on the current GMWM framework replacing the single WV vector with the weighted average $\sum_{i=1}^K w_i \hbnu_i$. By doing so, it can directly make use of existing starting-value algorithms while it is not immediate to apply the same algorithms for the MS-GMWM. For the latter estimator, one could for example use the AGMWM as a starting value for optimization but it would require extra orders of computations (i.e., AGMWM as a first step) and, since the AGMWM targets $\btheta^{\circ} = \mathbb{E}[\bvartheta_i]$, may not be a close enough starting value.

The choice of the matrix $\bm{\Omega}$ may not be completely obvious in the stochastic framework considered in this work. If one chooses an estimator $\widehat{\bm{\Omega}}$ then, in the standard single replicate setting, one can choose the inverse of the estimated covariance matrix of the empirical WV, i.e., $\widehat{\bm{\Omega}} := \hat{\mathbb{V}}[\hbnu_i]^{-1}$, or a diagonal matrix proportional to the latter. Since the matrix $\bm{\Omega}$ only affects the asymptotic efficiency of the resulting estimator and does not affect the consistency as long as it is positive definite, then one could choose the following matrix:
$$\widehat{\bm{\Omega}}_K := \sum^K_{i = 1} w_i \widehat{\bm{\Omega}}_i,$$
where $\widehat{\bm{\Omega}}_i$ represents the estimator for $\bm{\Omega}_0$ for the $i^{th}$ replicate. The weighted average of the matrices that would be used on the individual replicates is indeed a valid choice and, for this reason, is what is going to be used in the next applied sections.

As a final note, these multi-signal approaches would be valid also in a setting where the model parameters do not vary between replicates (i.e., the internal sensor model $G$ is a Dirac distribution) and would probably benefit from greater asymptotic efficiency, compared to methods applied to a single replicate, due to their averaging nature. Moreover, these new theoretical results can provide additional support to the improvement of inferential tools proposed for the near-stationary setting considered in this work. Indeed, \cite{bakalli2017computational} and \cite{radi2019multisignal} also suggested a multi-signal near-stationary test to determine whether the process parameters changed between replicates. In the latter studies, a bootstrap distribution is derived for the test statistic but the results on asymptotic normality of the estimators studied in this work could allow to make use of a more computationally efficient $\chi^2$-test for this purpose (however this is left for future research). In addition, the results on asymptotic normality of these methods allow for the use of time-dependent bootstrap methods (such as the moving block bootstrap) to estimate the corresponding asymptotic covariance matrices.

\section{Simulation Studies}
\label{sec:simu}

In this section we provide further support to the results presented in Sec. \ref{sec:ms} by studying the finite sample performance of the suggested solutions. In fact, based on these results, we only compare two of the considered solutions, namely the AGMWM and AWV (since the MS-GMWM is equivalent to the latter). To do so we perform simulation studies based on composite stochastic processes that often characterize the stochastic signals from inertial sensor measurements. The first is a relatively common example consisting in the sum of a White Noise (WN) process with a first-order AutoRegressive (AR1) process (the latter consisting in a re-parametrization of a Gauss-Markov process), while the second consists in a sum of these two processes with the addition of a Random Walk (RW). In this second simulation setting, we therefore also consider the presence of non-stationary processes in the error signals also commonly found in stochastic signal calibration.

\begin{figure}[ht!]
	\centering		
	\includegraphics[width = \linewidth]{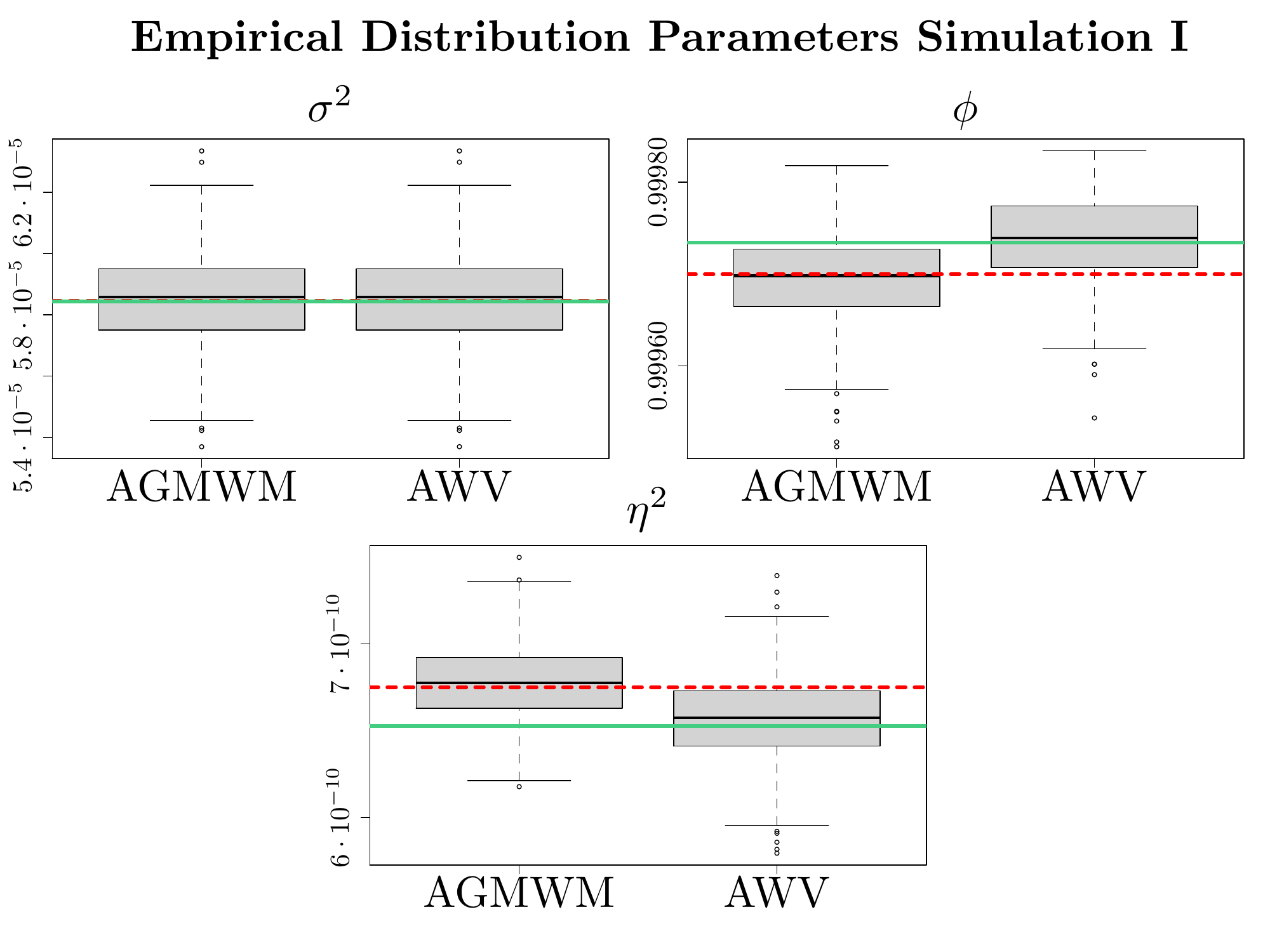}		
	\caption{Empirical distributions of the AGMWM (left boxplot) and AWV (right boxplot) for the parameters of the stochastic error model (WN + AR1) of Simulation I ($K = 6$ and $T = 10^6$). The red dashed line represents the parameter value $\btheta^\circ$, while the full green line represents $\btheta_0$.}
   \label{fig:boxplot1}
\end{figure}

In order to generate settings that closely resemble the WV plots that are observed in stochastic calibration sessions, we choose to represent the internal sensor model $G$ through independent and rescaled Beta distributions (i.e., each element of the parameter vector $\bvartheta_i$ comes from a separate rescaled Beta distribution). In addition, we choose to study the estimators in a setting where we observe $K = 6$ replicates which all have the same length, i.e., $T_i = T = 10^6$ (for all $i$), thereby delivering $J = 13$. Moreover, we choose $\bm{\Omega}$ by taking the average of the individual matrices for each replicate as discussed at the end of Sec. \ref{sec:ms}. We repeat this setting $B = 500$ times to investigate the empirical distribution of the estimators studied. Finally, to be able to understand if the estimators are targeting the correct values, we compute the value $\btheta_0$ via numerical simulations by minimizing $Q(\btheta)$ given in \eqref{eq:optimality} based on $K = 10^3$ values of $\bvartheta_i$ randomly generated from the chosen internal sensor model $G$, while $\btheta^\circ$ is computed for each element of $\bvartheta_i$ based on its corresponding distribution.
\begin{figure*}[!ht]
\begin{minipage}[b]{0.475\linewidth}
\centering
\includegraphics[width=\linewidth]{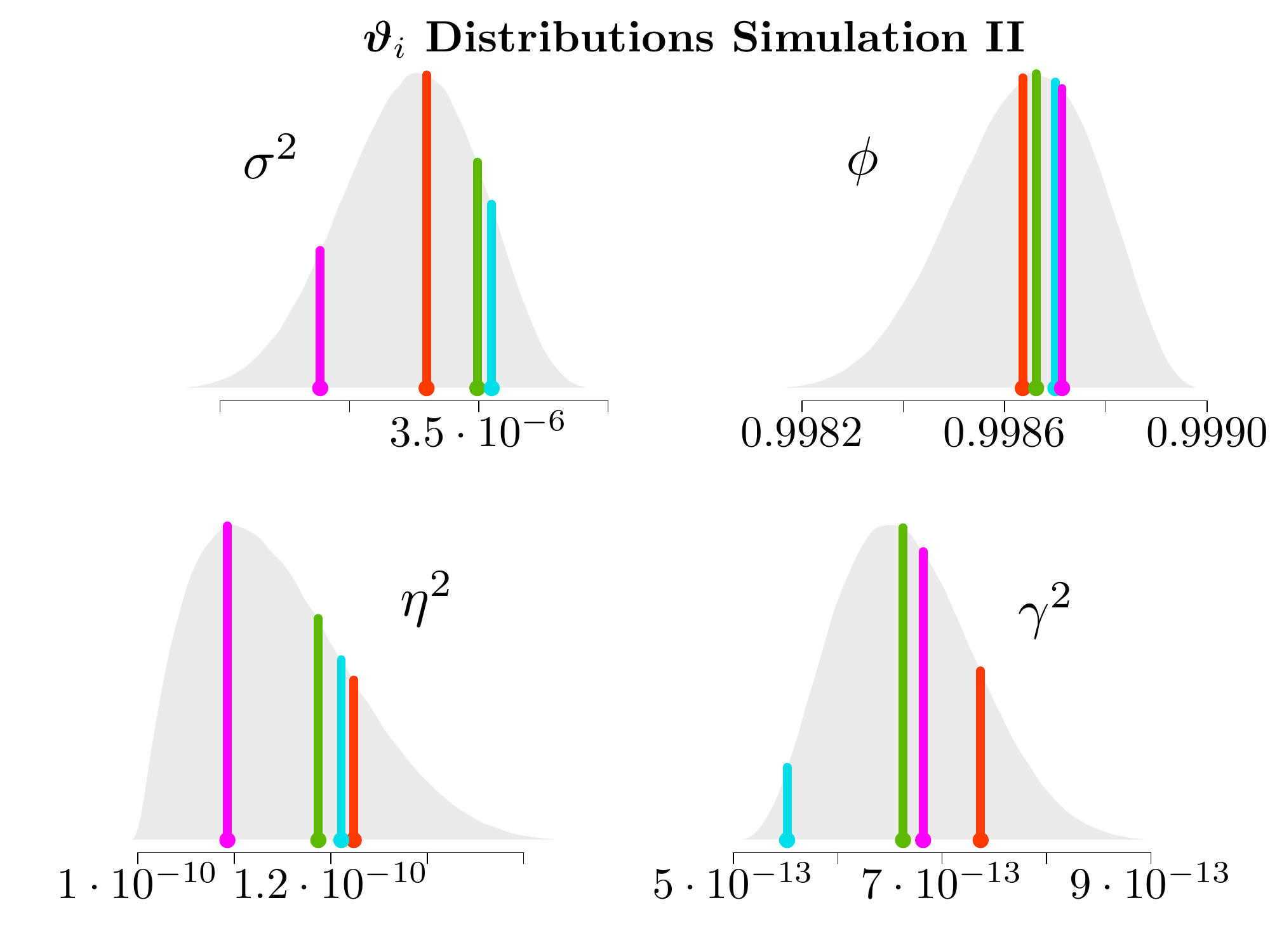}
\end{minipage}\hfill
\begin{minipage}[b]{0.475\linewidth}
\centering
\includegraphics[width=\linewidth]{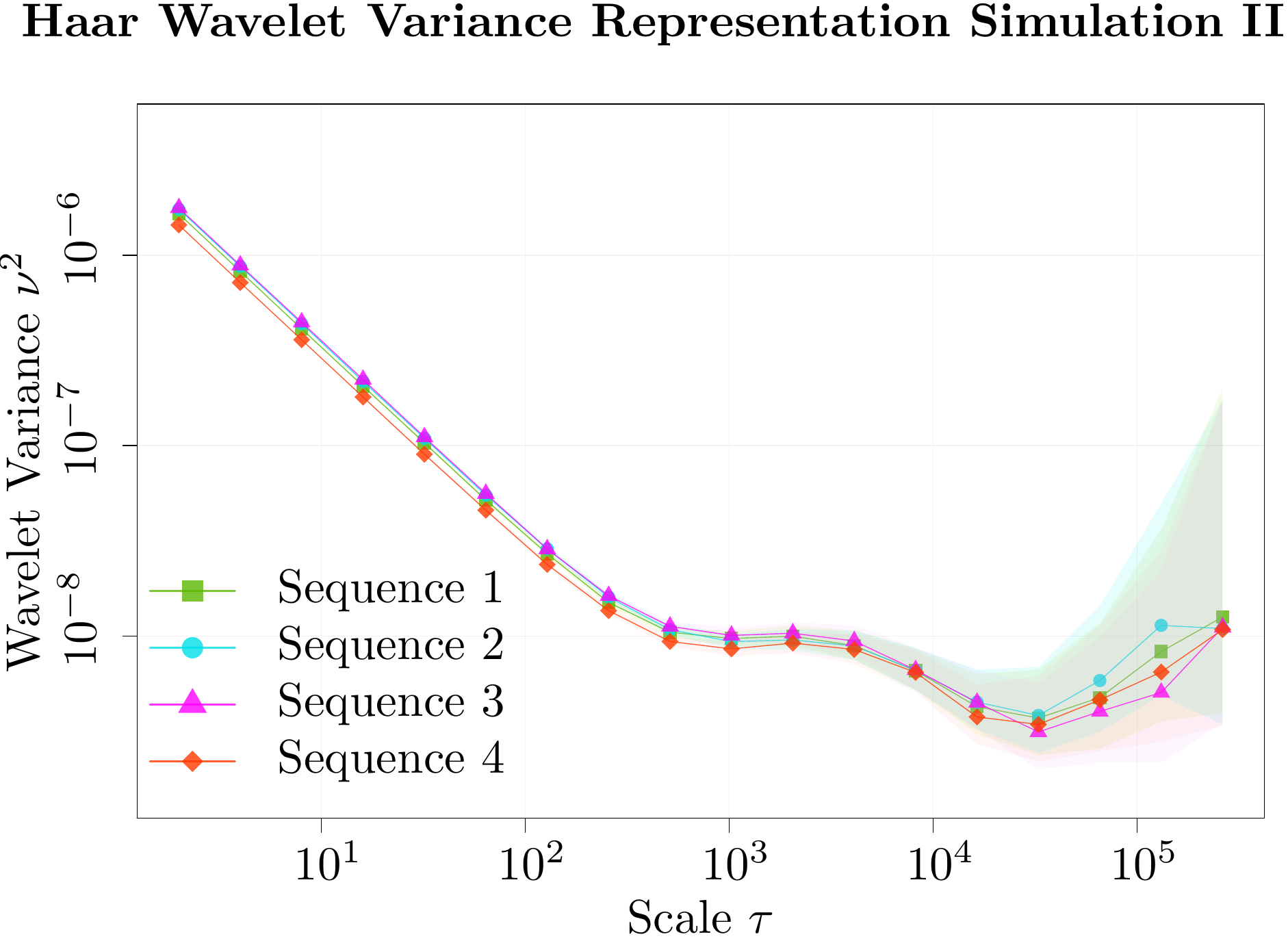}
\end{minipage}\hfill
\caption{Left: Marginal densities of the internal sensor model $G$ for the parameters $\sigma_i^2$, $\phi_i$, $\eta_i^2$ and $\gamma_i^2$ considered in the Simulation II setting (WN + AR1 + RW) with horizontal colored lines representing four randomly selected values from each density. Right: WV plots and 95\% confidence intervals for the empirical WV of the signals generated by the parameter values selected from the respective densities in the top part (each color in the top part corresponds to the color of the WV in the bottom part).}
\label{fig_wv_simu2}
\end{figure*}

\subsection{Simulation I}

\noindent For the first simulation, the parameter vector for the $i^{th}$ replicate is defined as follows $\bvartheta_i := [\sigma_i^2, \, \phi_i, \, \eta_i^2]$, where $\sigma_i^2$ represents the WN parameter, $\phi_i$ is the autoregressive parameter of the AR1, and $\eta_i^2$ is the innovation variance parameter of the AR1. In the near-stationary setting, we therefore have that $\bvartheta_i \sim G$ which we choose as follows:
\begin{itemize}
	\item $\sigma_i^2 = 4\cdot10^{-5} + Y_i^{(1)}(7\cdot10^{-5}  - 4\cdot10^{-5} )$, where \\$Y_i^{(1)} \sim \text{Beta}(8,5)$,
	\item $\phi_i = 9.99\cdot10^{-1} + Y_i^{(2)}(9.999\cdot10^{-1}  - 9.99\cdot10^{-1} )$, where $Y_i^{(2)} \sim \text{Beta}(7,2)$ ,
	\item $\eta_i^2 =  6\cdot10^{-10} + Y_i^{(3)}(8\cdot10^{-10}  - 6\cdot10^{-10} )$, where $Y_i^{(3)} \sim \text{Beta}(3,5)$.
\end{itemize}
An insight into the described simulation setting is given in Fig. \ref{fig_wv_simu1} where in the left part we can observe the rescaled Beta density functions (grey surfaces) from which we generate the respective parameter values that compose $\bvartheta_i$. Hence, the internal sensor model $G$ is the multivariate distribution composed of independent variables $\sigma_i^2$, $\phi_i$ and $\eta_i^2$. The vertical colored lines represent randomly sampled values for the parameters following their respective distributions where common colors indicate those values that were generated jointly to deliver four different values of $\bvartheta_i$. These colors are then used to represent the empirical WV computed on signals generated from each value of $\bvartheta_i$ which can be seen in the right part of Fig. \ref{fig_wv_simu1}. We can notice how the different WVs are extremely close at the first scales and then differ at the larger scales. This plot is very similar to those seen in many applied settings as shown in Sec. \ref{sec:case2}.
When applying the estimators to the setting described above, we observe the results shown in Fig. \ref{fig:boxplot1}. The red dashed line represents the true value of $\btheta^\circ$ and the full green line represents the (approximated) value of interest $\btheta_0$. The boxplots represent the empirical distribution of the estimated parameter values for the AGMWM (left boxplot) and AWV (right boxplot) respectively. While all boxplots appear to support the results on asymptotic normality of the estimators derived in Sec. \ref{sec:ms}, it can be observed that the corresponding elements of $\btheta^\circ$ and $\btheta_0$ appear to differ (especially for the AR1 process which is non-linear in the WV). As a result of these differences, it is also obvious to detect how the two estimators target these different quantities since the AGMWM is centered around the red dashed line ($\btheta^\circ$) and the AWV around the full green line ($\btheta_0$). This therefore supports the consistency results in Sec. \ref{sec:ms} which indeed state that these estimators target these respective quantities.

\subsection{Simulation II}
\noindent As mentioned at the start of this section, we perform a second simulation study in a similar way to the first one but, in this case, we add a RW process to the other two. This implies that the generated signals are non-stationary which is in fact the case for many stochastic error signals issued from inertial calibration sessions. For this simulation, we have that $\bvartheta_i := [\sigma_i^2, \, \phi_i, \, \eta_i^2, \, \gamma_i^2]$ where, in addition to the parameters specified in the previous simulation, $\gamma_i^2$ represents the parameter of the RW process. The internal sensor model is composed of the following random parameter distributions:
\begin{itemize}
	\item $\sigma_i^2 = 2\cdot10^{-6} + Y_i^{(1)}(4\cdot10^{-6}  - 2\cdot10^{-6} )$, where $Y_i^{(1)} \sim \text{Beta}(8,5)$,
	\item $\phi_i = 9.98\cdot10^{-1} + Y_i^{(2)}(9.99\cdot10^{-1}  - 9.98\cdot10^{-1} )$, where $Y_i^{(2)} \sim \text{Beta}(7,4)$,
	\item $\eta_i^2 =  1\cdot10^{-10} + Y_i^{(3)}(1.5\cdot10^{-10}  - 1\cdot10^{-10} )$, where $Y_i^{(3)} \sim \text{Beta}(3,5)$;
	\item $\gamma^2 = 0.5\cdot10^{-12} + Y_i^{(4)}(1\cdot10^{-12}  - 0.5\cdot10^{-12} )$, where $Y_i^{(4)} \sim \text{Beta}(4,8)$;
\end{itemize}
\begin{figure}[b]
	\centering		
	\includegraphics[width = \linewidth]{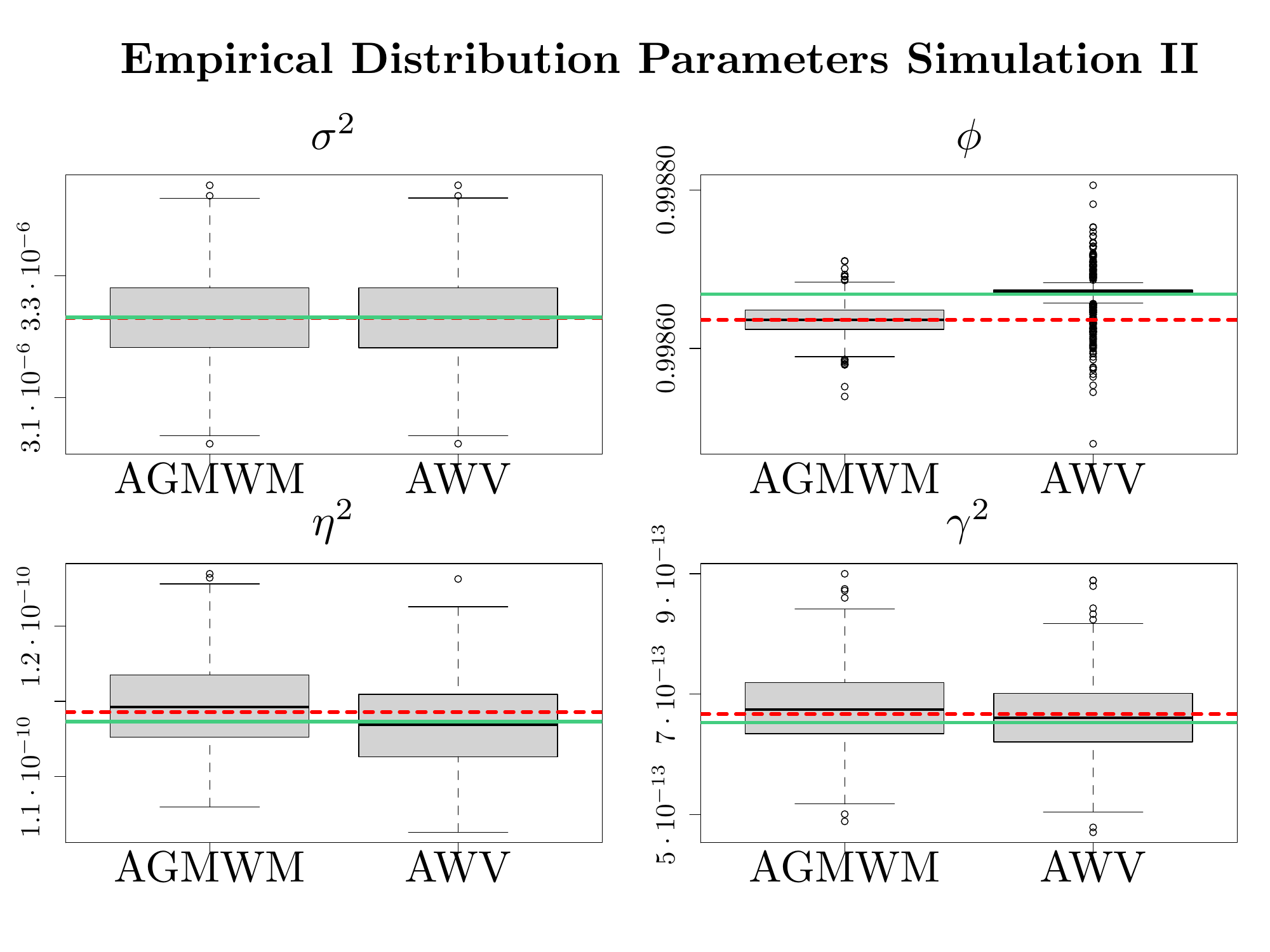}		
	\caption{Empirical distributions of the AGMWM (left boxplot) and AWV (right boxplot) for the parameters of the stochastic error model (WN + AR1 + RW) of Simulation II ($K = 6$ and $T = 10^6$). The red dashed line represents the parameter value $\btheta^\circ$, while the full green line represents $\btheta_0$.}
   \label{fig:boxplot2}
\end{figure}

Again, to give a visual support to the setting of this simulation, we provide an example of the parameter distributions (grey areas) along with four randomly sampled values for $\bvartheta_i$ represented by the four different colors in Fig.~\ref{fig_wv_simu2} left panels. Also in this case, it is possible to notice how the empirical WV generated from these different parameter values differ across the scales and, it can also be seen how some can be significantly different from the others at the first scales as highlighted by the non-overlapping confidence intervals of the respective WVs (shaded areas in the WV plot).
In a similar manner to the first simulation we represent the results when applying the two considered estimators to this near-stationary setting. These results, represented in Fig. \ref{fig:boxplot2}, confirm the conclusions made in the first simulation where both estimators appear normally distributed and both target their respective values of reference, i.e., $\btheta^\circ$ for the AGMWM and $\btheta_0$ for the AWV.
Having given empirical support to the conclusions made in Sec. \ref{sec:ms}, we now study how these conclusions deliver advantages in applied cases. In the next section, we therefore study the results in terms of navigation performance when using the AWV estimator which targets the value of interest $\btheta_0$. 

\section{Case Study - Impact on Navigation}
\label{sec:case2}
The purpose of this section is to compare how navigation performances change when estimating stochastic models for the inertial sensors using a single replicate of the calibration data (as it is currently done), based on the GMWM, or using all replicates jointly based on the AWV estimator put forward in this work.

We collect static measurements from a Bosch Sensortec BMI085 6-Axis IMU \footnote{The Bosch Sensortec BMI085: \url{https://www.bosch-sensortec.com/products/motion-sensors/imus/bmi085/}}, a low-cost MEMs IMU (\hbox{$< 5$} USD per unit, when purchased in volumes) for navigation applications, e.g., in UAVs. Such an inertial module combines a 3-axis gyroscope and a 3-axis accelerometer. We collect $K=16$ replicates of sensor data in static conditions at $20$~${}^\circ$C in a temperature controlled chamber, each one  lasting $12$ hours. Since the sensor is static, the acquired data consists of samples of the noise processes only. 
The sensor runs at a frequency of $200$ Hz, thus each error signal contains approximately $8.5$ million sample points. We focus on the error signals from the X-axis gyroscope and accelerometer. 

To identify the error process we visually analyse the empirical WV of eight sequences that we consider for training purposes (i.e. used to estimate the model parameters), while leaving the remaining eight for validation, as discussed later on. The empirical WV of the training sequences are shown in Fig.~\ref{fig:emp_hexa}. We observe that the considered devices are characterized by a non-negligible bias-instability, as it can be seen from the relatively flat part of the WV at the larger scales. This behaviour is common in low-cost inertial sensors and it is typically modeled with a sum of first order auto-regressive processes~(AR1), or equivalently, first order Gauss-Markov processes, as suggested for example in~\cite{yuksel2011notes, ieee1998ieee}. We find that three AR1 processes are well suited to model each training sequence for the gyroscopes, and four for the accelerometers, respectively. We note that in both cases one of such AR1 processes always has a very short correlation time, far smaller than $1$ s. This process models the intrinsic bandwidth limitation of the sensor (visible in the elbow at the first two scales of the WV) and is typically replaced with a white noise (an Angular/Velocity Random Walk) in practice. We estimate one model separately on each sequence in the training set, obtaining models $\mathcal{M}_i$, with $i \in [1, ..., 8]$. Next, we apply the AWV method proposed in this work employing all eight training sequences together, obtaining the model denoted as $\mathcal{M}_\text{MS}$. The estimated training models $\mathcal {M}_i$ appear to adequately fit the empirical WV of their respective training sequence, thus supporting the choice of the general model  (an example consisting in the first training sequence is provided in Fig.~\ref{fig:fit_seq_1}). The fits for each sequence are given in Appendix~{\color{red}A}  Fig. {\color{red}A.1}  and  {\color{red}A.2}. 

Each fitted model lies within the confidence intervals of the empirical WV. Considering these representations, it is straightforward to detect differences in the models fit to the signals via the individual and joint approaches. Given this, in order to confirm whether to use a single replicate or a multi-signal approach we perform the near-stationarity test put forward in \cite{bakalli2017computational} by simulating 100 bootstrap replicates under the estimated $F_{\hbtheta^{\dagger}}$ which, keeping in mind the discrete nature of the bootstrapped test statistic, gives us a zero p-value thereby allowing us to reject the null hypothesis that all replicates are issued from the same data-generating process with $\bvartheta_i = \theta_0$ for all $i$ (i.e. $G$ is a Dirac point mass distribution).  The estimated parameters of the models $\mathcal{M}_i$ are included in Appendix~\ref{sec:distrib}, Fig.\ref{fig:params_acc} and \ref{fig:params_gyro}. We note that a substantial variability can be observed within the latter fits and that, as expected, the parameters obtained with the AWV method do not correspond to their mean.

We investigate the navigation performance on the $8+1$ different models. The estimated stochastic models are used to configure an Extended Kalman Filter (EKF) for INS/GNSS navigation~\cite{titterton2004strapdown}. This filter fuses inertial and GNSS readings, leveraging on the provided stochastic models, to estimate the vehicle navigation states (position, velocity and orientation). It allows us to compare the performance of the different models available for the inertial sensor in terms of position and orientation errors as well as consistency of the confidence intervals for the navigation states within a realistic navigation scenario. We consider a ground-truth trajectory typical of a small fixed-wing Unmanned Aerial Vehicle~(UAV) performing an aerial mapping mission. A $30$s GNSS outage period is considered after $9.5$ minutes. All the true kinematic properties of the sensors are known (position, velocity, etc.) from the reference trajectory and they are used to generate synthetic, noise-free sensor readings for both the inertial and the GNSS sensors. Realistic noisy readings are then generated for the inertial sensors by adding samples from the noise replicates collected during static acquisitions to the synthetic noise-free readings. Here, we employ the remaining eight static data sequences we collected and that were never used in the previously described stochastic calibration step. As for the GNSS readings, the added noise is WN with standard deviation $2.5$ cm, which corresponds to the assumed uncertainty carrier-phase differential of GNSS typically employed in mapping missions.
\begin{figure}[t]
    \centering
    \includegraphics[width=\linewidth]{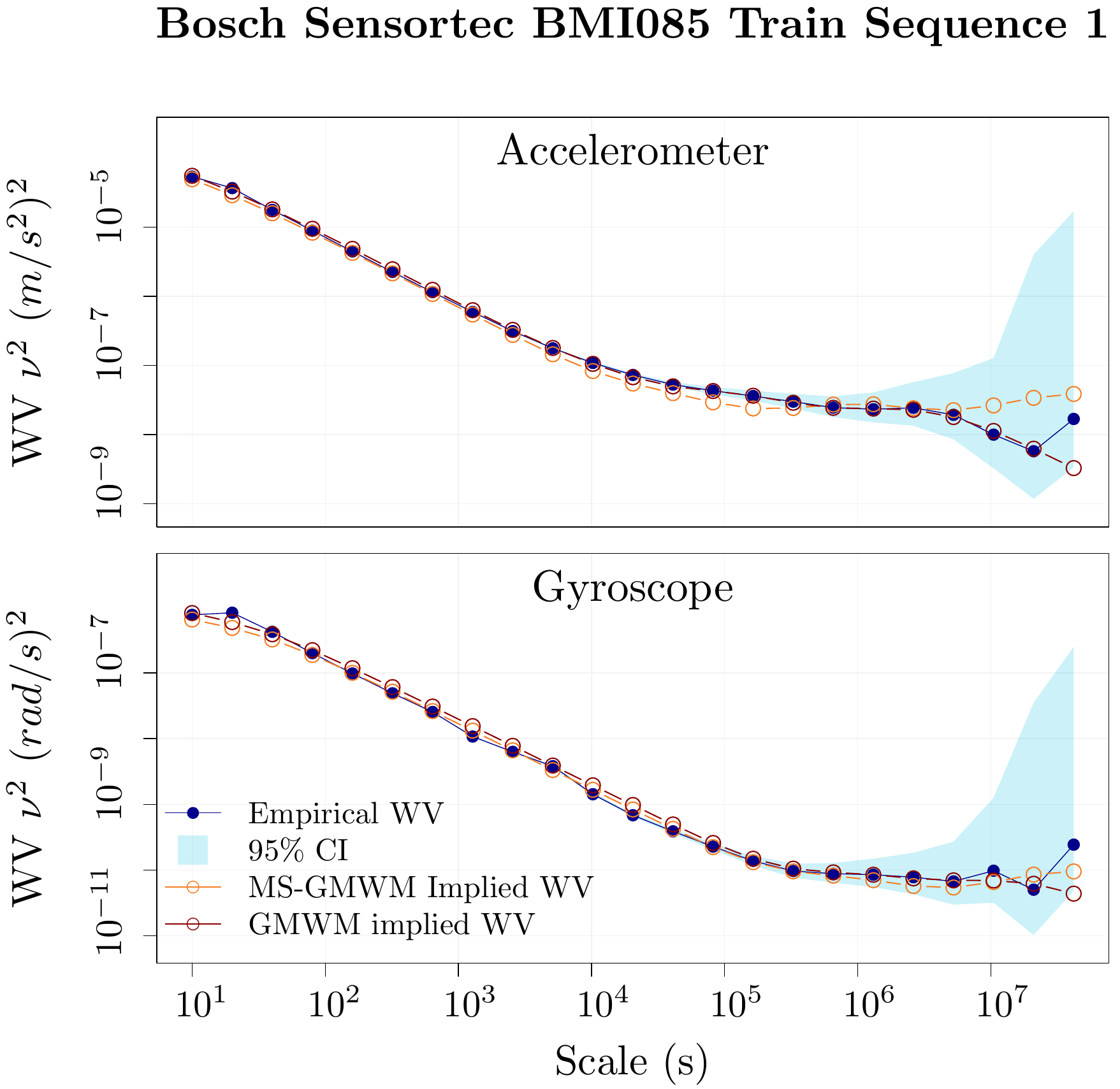}
    \caption{Empirical WV (blue doted line) for the first training sequence of accelerometer and gyroscope of a Bosch Sensortec BMI085 6-Axis IMU and their respective $95\%$ confidence intervals (blue shaded area). Red dotted lines represent the implied WV from the individual solution of the GMWM on this first sequence, while orange dotted lines represent the WV implied by the MS-GMWM computed trough the AWV.  }
    \label{fig:fit_seq_1}
\end{figure}

A forward navigation solution is computed using an EKF from the noisy sensor readings. We consider $9 \times 8 = 72$ different cases in which the EKF is configured to use one of the $8+1 = 9$ model sets fitted on the static acquisition replicates, while the noise data corrupting inertial readings comes from one of the  eight different static acquisition  sequences kept for validation, each time considering a different, continuous chunk of data.

The $250$ solutions for each case are aggregated and compared in terms of relative position and orientation error and consistency of the confidence intervals computed by the EKF: we compute $50$\% confidence intervals (approximately corresponding to the common choice of $\pm \sigma$ intervals) from the navigation state covariance matrix estimated by the EKF and we count how many times the true navigation states (from the reference trajectory) fall within such confidence intervals.  Note that it is equivalent to check whether the Average Normalised Estimation Error Squared~(ANEES), as defined for example in~\cite[Chapter 3.7.4]{bar2004estimation}, falls within its expected bounds, and it allows to quantify whether the employed stochastic models for the inertial errors lead to a over- or under-confident estimation of the navigation state uncertainty. The position and orientation error and the coverage metrics are evaluated each $0.5$s in the last $15$s of the GNSS outage period to better highlight their evolution when the navigation filter works in standalone mode, e.g., relying only on inertial data (represented in Fig.~\ref{fig:metrics_evaluation}). The results are presented in Fig.~\ref{fig:navigation}.

\begin{figure}[t]
    \centering
    \includegraphics[width=\linewidth]{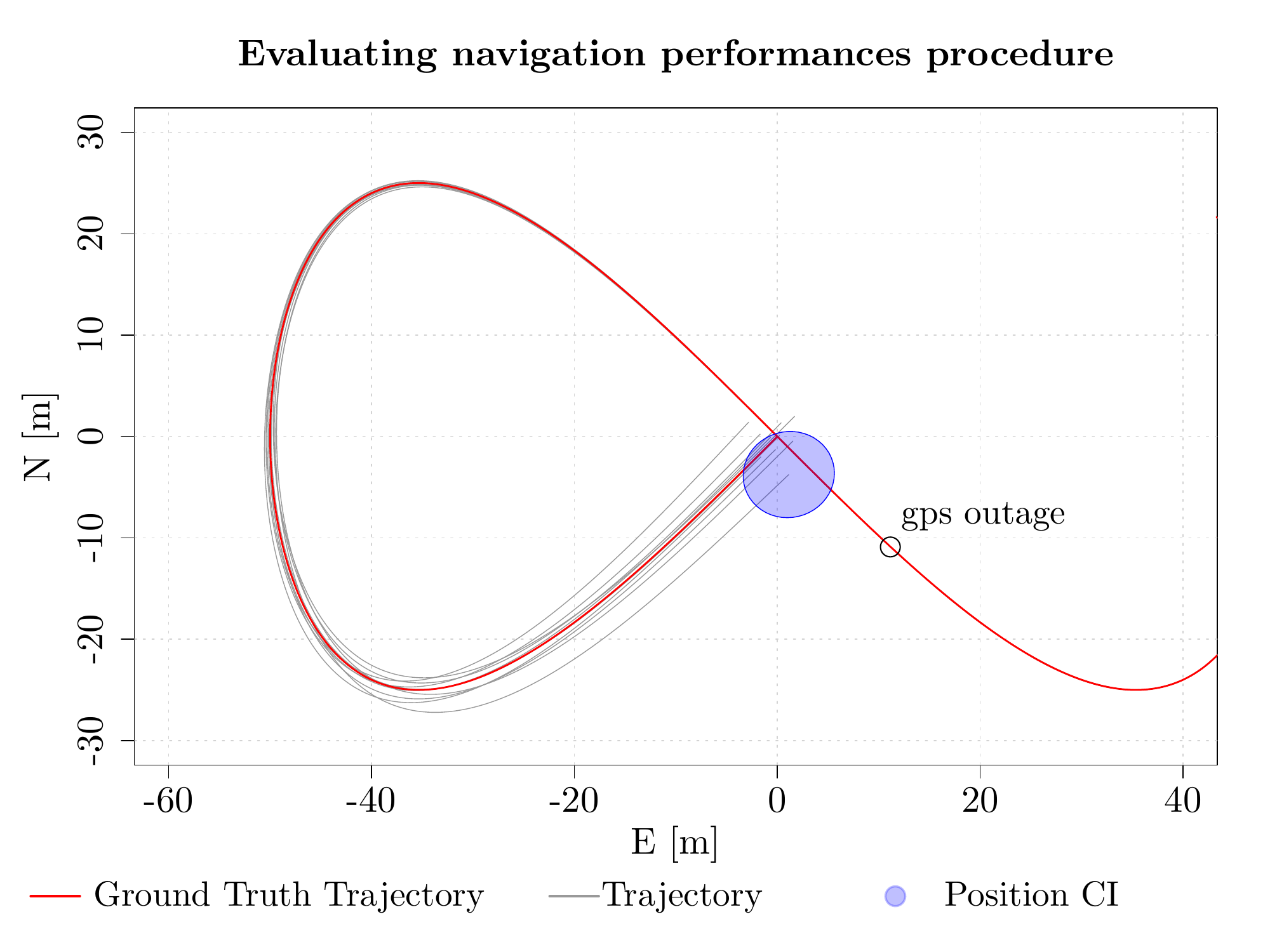}
    \caption{Procedure for evaluating navigation performances. An EKF estimates the trajectory of an UAV in $250$ Monte-Carlo runs (a few of them are represented through the thin dark lines). The GNSS position and velocity information are no longer available after the marked point (gps outage). The position and orientation error and the coverage of the uncertainty of the navigation states (blue circle), as estimated with by the EKF, are assessed based on the ground truth trajectory (red line) during the last $15$s of the GNSS outage period. }
    \label{fig:metrics_evaluation}
\end{figure}

It is possible to see that the differences in position and orientation error, computed in percentage with respect to the best performing model, vary up to $5$ \% depending on which stochastic model is selected for the inertial sensor.  These differences may seem small, but attitude quality improvement is proportional to the square (or even the cube) of the IMU size and weight (as well as cost). The differences in coverage are much more significant: when computing a confidence interval for position and orientation with level $\alpha = 0.5$ (50\%), we find that the empirical coverage of certain models fit on a single sequence, e.g., $\mathcal{M}_4$ and $\mathcal{M}_8$, is as low as $10$ \% or as high as $90$ \% in some cases. This implies that, when configured with such models, the EKF is largely over- or under-confident in the estimation of the uncertainty of the navigation states. Even though the actual errors in such states remain relatively small, the quantification of their uncertainty is substantially unreliable which prevents, for example, proper decision making in safety-critical navigation applications, or consistent information fusion in more complex scenarios such as simultaneous localisation and mapping, where further sensor information (e.g., from cameras) need to be taken into account. On the other hand, the model estimated with one of the methods put forward and studied in this work, $\mathcal{M}_{\text{MS}}$, achieves almost optimal position and orientation performances, while at the same time providing a reliable and correct uncertainty quantification of the position and orientation estimates. We remark that \emph{by chance} one single sequence may lead to the estimation of a stochastic model which performs well in practice, but at the same time the opposite may hold, for example if training sequences $3$, $4$, or $8$ were to be selected. These results indicate that the AWV (or a multi-signal method) can deliver a more robust (stable) estimation of the stochastic models that underlie inertial sensor measurement errors, compensating for the intrinsic variability of the single realizations of calibration data.

\begin{figure*}
    \centering
    \includegraphics[width=\linewidth]{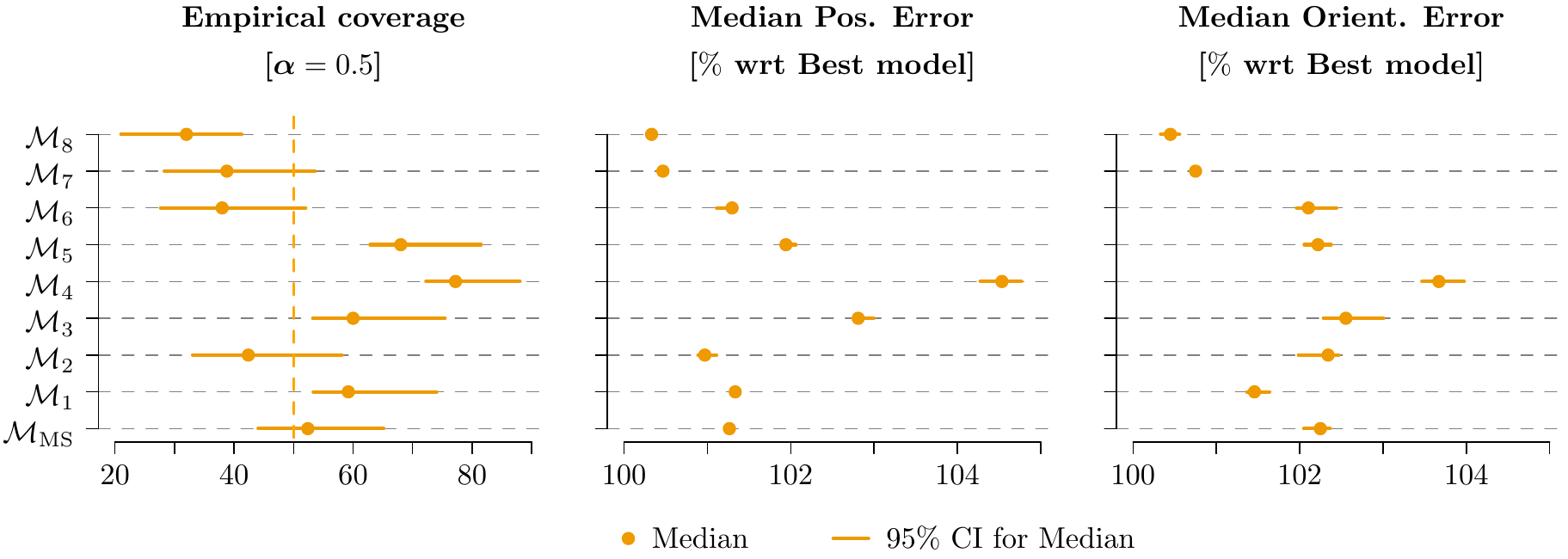}
    \caption{Empirical coverage of the $50$\% confidence intervals derived from the EKF covariance matrices, median position and orientation errors achieved by using each estimated model to predict the error on all replicates, $\mathcal{M}_i$ represents the model estimated on replicate $i$ and $\mathcal{M}_{\text{MS}}$ represents the model estimated via the multi-signal AWV. The results are expressed in percentage with respect the best performing model on one specific static acquisition.}
    \label{fig:navigation}
\end{figure*}

\section{Conclusions}
\label{sec:conclusion}

In this work, we studied methods and delivered further evidence for the need of a multi-signal approach when dealing with inertial sensor calibration. Indeed, in many practical settings, one can observe a near-stationary behavior of replicate IMU stochastic error signals which needs to be taken into account when performing estimation for model selection and construction of accurate navigation filters. Having compared different existing and new approaches to address this problem, we determined their asymptotic properties and their common features which were empirically supported in controlled simulation settings as well as in applied case study scenarios. In the latter case, this work also highlighted how the use of a single replicate to perform stochastic calibration may be a sub-optimal choice and confirmed that a multi-signal solution is the most appropriate in such settings. As a result of this work, it is now possible to select the most appropriate multi-signal calibration approach according to the goal of interest and consequently achieve improved navigation performance both in terms of accuracy as well as in terms of uncertainty quantification during navigation. Finally, this study can extend to all approaches based on moment-matching (e.g. Generalized Methods of Moments) beyond the WV and IMU calibration.

\section*{Acknowledgment}
We are grateful to M.-P. Victoria-Feser for her helpful comments. This work was supported in part by the SNSF Grant $\#100018-182582$, in part by the SNSF Professorships Grant $\#176843$ and by the Innosuisse-Boomerang Grant $\#37308.1$ IP-ENG.



\bibliographystyle{unsrt}
\bibliography{ref}

\end{document}